\documentclass[a4paper,UKenglish,cleveref, autoref]{lipics-v2021}

\hideLIPIcs

\usepackage{amsmath}
\usepackage{amsfonts}
\usepackage{amssymb}
\usepackage{amsthm}
\usepackage{stmaryrd}
\usepackage{enumitem}
\usepackage{xfrac}

\usepackage{cite}
\usepackage{hyperref}

\usepackage{graphicx}
\graphicspath{ {} }

\usepackage[]{mdframed}

\usepackage{verbatim}
\usepackage{multirow}

\usepackage{wrapfig}

\newcommand{\mystrut}{\raisebox{0ex}[0.37cm]{}}

\usepackage[appendix=append,bibliography=common]{apxproof}

\usepackage{tikz}
\usetikzlibrary{automata, positioning, arrows, shapes, snakes}

\newtheoremrep{theorem}{Theorem}
\newtheoremrep{definition}[theorem]{Definition}
\newtheoremrep{lemma}[theorem]{Lemma}
\newtheoremrep{proposition}[theorem]{Proposition}
\newtheoremrep{corollary}[theorem]{Corollary}
\newtheoremrep{example}[theorem]{Example}
\newtheoremrep{remark}[theorem]{Remark}

\title{Witnesses for Fixpoint Games on Lattices}

\author{Barbara K\"onig}{Universit\"at Duisburg-Essen, Germany}{barbara_koenig@uni-due.de}{https://orcid.org/0000-0002-4193-2889}{}
\author{Karla Messing}{Universit\"at Duisburg-Essen, Germany}{karla.messing@uni-due.de}{https://orcid.org/0009-0003-1019-6449}{}

\authorrunning{B. K\"onig and K.Messing} 

\Copyright{Barbara K\"onig and Karla Messing} 



\nolinenumbers 

\relatedversiondetails[cite={km:witnesses-fixpoint-games}]{Conference
  Version}{https://doi.org/10.4230/LIPIcs.ICALP.2026.182}

\funding{The authors were supported by the Deutsche
  Forschungsgemeinschaft (DFG, German Research Foundation) -- project
  number 434050016 (SpeQt)}

\DeclareMathOperator{\degr}{\mathsf{deg}}
\DeclareMathOperator{\cdegr}{\mathsf{cdeg}}

\DeclareMathOperator{\lo}{\mathsf{log}}
\DeclareMathOperator{\be}{\mathsf{beh}}
\DeclareMathOperator{\cl}{\mathsf{cl}}

\newcommand{\lt}[1]{\mathbb{#1}}
\newcommand{\du}[1]{\dot{#1}}

\newcommand{\pr}[1]{#1}
\newcommand{\mb}[1]{M_{\lt{#1}}}
\newcommand{\jb}[1]{J_{\lt{#1}}}

\newcommand{\strat}{S}
\newcommand{\auxp}[1]{W_{p,#1}}
\newcommand{\auxd}[1]{W_{d,#1}}
\newcommand{\aps}{W_p}
\newcommand{\apc}{\auxp{\lo}}
\newcommand{\ads}{W_d}
\newcommand{\adc}{\auxd{\lo}}
\newcommand{\adwa}{Z}

\DeclareMathOperator{\wal}{\triangleleft\!\triangleleft}
\DeclareMathOperator{\wag}{\triangleright\!\triangleright}

\newcommand\twoheaduparrow{\mathrel{\rotatebox[origin=c]{90}{$\twoheadrightarrow$}}}
\newcommand\wadownarrow{\mathrel{\rotatebox[origin=c]{90}{$\wal$}}}

\renewcommand{\epsilon}{\varepsilon}
\renewcommand{\phi}{\varphi}

\begin{document}

\maketitle

\begin{abstract}
  We construct witnesses that can be used to derive strategies in
  fixpoint games and provide proof that the least fixpoint of a
  function is either \emph{above} or \emph{not below} some given
  bound. We rely on a lattice-theoretical approach, including a Galois
  connection that connects a lattice representing the ``logic
  universe'', where the witness lives, with another lattice
  representing the ``behaviour universe'', over which the function is
  defined. In fact we consider two types of games -- primal and dual
  games -- and in both cases show how to derive winning strategies in
  the game from witnesses and construct witnesses from strategies. The
  two games differ wrt.\ their rules and the choice of basis of the
  lattice.

  The theory can be instantiated to well-known examples: in particular
  we compare with the construction of distinguishing formulas in
  standard bisimilarity and behavioural metrics for probabilistic
  systems. As a new case study we consider witnesses for certifying
  lower bounds for the termination probability for Markov chains.
\end{abstract}

\section{Introduction} 
In concurrency theory, there are many concepts that arise as least or
greatest fixpoints of functions, for instance termination probability
in Markov chains \cite{bk:principles-mc,gs:markov-chains}, values for
Markov decision processes \cite{b:markovian-decision-process} and
simple stochastic games \cite{c:complexity-stochastic-game},
bisimilarity \cite{s:bisimulation-coinduction,stirling99} or
behavioural metrics
\cite{dgjp:metrics-labelled-markov,cgt:logical-bisim-metrics,afs:linear-branching-metrics,flt:quantitative-spectrum,bw:behavioural-pseudometric,bbkk:coalgebraic-behavioral-metrics}. Limiting
the discussion to least fixpoints (which can always be done without
loss of generality via dualization), such a fixpoint characterization
can also provide upper bounds via pre-fixpoints. However here we are
instead interested in lower bounds, respectively in showing that upper
bounds do not hold. In the case of bisimilarity, the first task is
equivalent to certifying that two states \emph{are} bisimilar, while
the second means certifying that they are \emph{not}. The latter
notion is also known as apartness and has recently garnered increased
attention \cite{gj:apartness-bisimulation,tbkr:proving-beh-apartness}.

One way to witness non-bisimilarity is via so-called distinguishing
formulas \cite{c:automatically-explaining-bisim}. The Hennessy-Milner
theorem \cite{hm:hm-logic} guarantees that in a finitely-branching
transition system two states are bisimilar if and only if they
satisfy the same formulas of the Hennessy-Milner logic, a form of
modal logic \cite{stirling99}. Hence, if two states are non-bisimilar,
there must be a formula distinguishing them. A similar notion of
distinguishing formulas has been studied in the quantitative setting
of behavioural metrics \cite{RadyBreugelExplProbBisim}, where such a
formula certifies that the distance of two states is above some bound
(or a sequence of formulas is constructed that certifies the distance
in the limit). Distinguishing formulas can in particular be used for
purposes of explainability, e.g. explaining why two states are
non-bisimilar as opposed to just stating this fact. Although their
existence follows directly from the Hennessy-Milner theorem, their
construction in general does not and has to be described separately.

Our aim is to lift the notion of distinguishing formulas -- here
called \emph{witnesses} for greater generality -- to a
lattice-theoretic setting, where we assume a Galois connection
connecting the ``logic universe'' with the ``behaviour universe''. (In
the case of bisimilarity logical formulas live in the logic universe,
while bisimulation relations live in the behaviour universe.)

Apart from (least) fixpoints and logics, there is a third equivalent
characterization that relies on games, that also provides us with the
concept of (winning) strategies. In particular, we consider fixpoint
games over continuous lattices \cite{POPL19FPGames} and we will show
how to translate witnesses into game strategies and vice versa. In
fact, there are two games for least fixpoints, a primal game and a
dual game, where the dual game is derived by taking the game for the
greatest fixpoint and dualizing the order. In both games the
existential player $\exists$ makes the first move, followed by the
universal player $\forall$. In the primal game the witness provides
the strategy for the existential player, while in the dual game we
obtain the strategy of the universal player, hence both cases have to
be treated separately. Both versions also have different
requirements: in the primal case we assume continuity of the behaviour
function and the lattice under consideration must have a basis
consisting of irreducibles. In the dual case the function has to be
co-proper.

Our results and constructions assume a continuous (resp.
co-continuous) lattice and use the so-called way-below and
way-above relations as well as notions from Scott topology
\cite{CompContLattices}. In particular we focus on constructing
\emph{finitary} formulas corresponding to \emph{finitary}
strategies.
As applications we will consider bisimilarity as well as behavioural
metrics for probabilistic systems. As a new case study we will show
how to provide witnesses for certifying lower bounds for the
termination probability for Markov chains.

The theory provides a framework and guideline for witnesses and
distinguishing formulas, in particular we provide generic algorithms
for deriving strategies from witnesses and witness
construction. We believe that viewing this construction in a more
general framework is valuable, in particular since its usefulness
seems to extend beyond behavioural equivalences and metrics.

\subparagraph*{Setup:} Assume a monotone function
$\be \colon \lt{B}\to \lt{B}$ (also called \emph{behaviour function})
on a complete lattice $(\lt{B},\sqsubseteq)$ and we are interested in
its least fixpoint $\mu \be$. Then -- due to the Knaster-Tarski
theorem \cite{t:lattice-fixed-point} -- in order to show that a given
lattice element $b\in\lt{B}$ is above $\mu \be$
($\mu \be \sqsubseteq b$), it is sufficient to find $p\in \lt{B}$ that
satisfies $\be(p)\sqsubseteq p\sqsubseteq b$ (where the first
inequality states that $p$ is a pre-fixpoint), from which we can
immediately deduce $\mu \be\sqsubseteq b$.

Our interest however lies in certifying the negated statement via a
witness, that is $\mu \be\not\sqsubseteq b$ (or alternatively
$b\ll \mu\be$ where $\ll$ is the way-below order). For the application
examples this means to show that two states are not bisimilar or that
the behavioural distance of two states is strictly larger than some
bound. In this case it helps to consider a fixpoint game
\cite{POPL19FPGames} where the aim of one player ($\exists$, defender)
is to prove that the inequality $\mu \be\sqsubseteq b$ holds, while it
is the aim of the opponent ($\forall$, attacker) to show that it does
not.

In particular, we want to represent the strategy of the attacker. One
way to do this is to look for a witness (also known as
distinguishing formula
\cite{c:automatically-explaining-bisim,RadyBreugelExplProbBisim}),
from which such a strategy can be derived. In lattice-theoretic terms,
such witnesses live in another lattice -- denoted by $\lt{L}$ -- that
is related to $\lt{B}$ via a Galois connection
$\alpha\colon \lt{L}\to\lt{B}$, $\gamma\colon \lt{B}\to \lt{L}$. We
assume the existence of a \emph{logic function}
$\lo\colon \lt{L}\to \lt{L}$ such that the left adjoint of the Galois
connection maps its least fixpoint to the least fixpoint of the
behaviour function ($\alpha(\mu \lo) = \mu \be$). Now the aim is to
look for $\ell\sqsubseteq \mu\lo$ that induces a strategy for the
fixpoint game for $\be$ on $\lt{B}$.

In fact, it turns out that there are two games one can consider on
$\lt{B}$, a primal game and a dual game. Depending on which game
is chosen, this influences the notion of witnesses and the translation
between witnesses and strategies.

\section{Preliminaries}

We recall some basic definitions about lattices, Galois connections,
the Hennessy-Milner framework based on Galois connections as
introduced in \cite{HMLGalois}, as well as Scott topology
\cite{CompContLattices}.

\subparagraph*{Partial Orders and Lattices:}

A \emph{complete lattice} $(\lt{L},\sqsubseteq)$ consists of a set
$\lt{L}$ and a partial order $\sqsubseteq$ defined on $\lt{L}$
such that every subset $Y\subseteq\lt{L}$ has a least upper bound
$\bigsqcup Y$ and a greatest lower bound $\bigsqcap Y$. The bottom and
top elements of $\lt{L}$ are denoted by $\bot$, $\top$ respectively

By Knaster-Tarski \cite{t:lattice-fixed-point} every monotone function
$f\colon\lt{L}\to\lt{L}$ has a least fixpoint $\mu f$ and a greatest
fixpoint $\nu f$. The least fixpoint can be obtained by Kleene
iteration over the ordinals, i.e.,
$\mu f = \bigsqcup_{i\in\mathsf{Ord}} f^i(\bot)$, dually for the
greatest fixpoint $\nu f$.

We recall some notions on (co-)continuous lattices
from~\cite{CompContLattices}. 

Let $(\lt{L},\sqsubseteq)$ be a lattice. A subset $D\subseteq \lt{L}$
is \emph{directed} if $D\neq \emptyset$ and every finite subset of $D$
has an upper bound in $D$.
For $\ell,\ell'\in \lt{L}$ we say that $\ell$ is \emph{way-below}
$\ell'$ ($\ell \ll \ell'$) iff for all directed subsets
$D\subseteq \lt{L}$, the relation $\ell'\sqsubseteq \bigsqcup D$
implies the existence of $d\in D$ with $\ell\sqsubseteq d$.
The lattice $\lt{L}$ is called continuous if for all $\ell\in \lt{L}$,
$\ell=\bigsqcup \{\ell'\in \lt{L} \mid \ell' \ll \ell\}$.
(Note that on a continuous lattice, whenever $\ell\ll\ell'$ 
and $\ell'\sqsubseteq \bigsqcup D$ for 
a directed set $D$, there even exists $d\in D$ such that $\ell\ll d$
\cite{CompContLattices}.) 

We also use the dual notions and call a set $F$ filtered if
$F\not = \emptyset$ and every finite subset of $F$ has a lower bound
in $F$; and say that $\ell$ is \emph{way-above} $\ell'$
($\ell'\wal \ell$) iff for all filtered subsets $F\subseteq \lt{L}$,
the relation $\ell'\sqsupseteq \bigsqcap F$ implies the existence of
$f\in F$ such that $\ell\sqsupseteq f$.\footnote{Note that in this
  terminology, ``way-above'' is not simply the inverse of ``way
  below''.}  The lattice $\lt{L}$ is called co-continuous if for all
$\ell\in \lt{L}$,
$\ell=\bigsqcap \{\ell'\in \lt{L} \mid \ell' \wag \ell\}$.

\smallskip

For $\ell\in \lt{L}$, we write
$\mathop{\uparrow} \ell = \{\ell'\in \lt{L} \mid \ell\sqsubseteq \ell' \}$,
$\mathop{\downarrow} \ell = \{\ell'\in \lt{L} \mid \ell\sqsupseteq \ell' \}$,
$\mathop{\twoheaduparrow} \ell = \{\ell'\in \lt{L} \mid \ell\ll \ell' \}$ and
$\mathop{\wadownarrow} \ell = \{\ell'\in \lt{L} \mid \ell\wag \ell' \}$.  The
set of ordinals will be denoted by $\mathsf{Ord}$.

\begin{example}
  \label{ex:way-below-above}
  Let $X$ be a set. For the powerset lattice
  $(\mathcal{P}(X),\subseteq)$, a set $S\subseteq X$ is way-below
  $S'\subseteq X$ ($S\ll S'$) if $S\subseteq S'$ and $S$ finite. Since
  every set arises as the union of finite sets, powerset lattices are
  continuous.  The sets $S,S'$ are in the way-above relation
  ($S\wal S'$) iff $S\subseteq S'$ and $S'$ is co-finite (complement
  wrt. the superset $X$ of a finite set). Powerset lattices are
  co-continuous, since every set is the intersection of co-finite
  sets.

  For the lattice $([0,1],\le)$ it holds that two elements
  $r,s\in[0,1]$ are in the way-below relation ($r\ll s$) if both are
  $0$ or $r<s$. Every real number is the supremum of strictly smaller
  numbers, hence $[0,1]$ is a continuous lattice as
  well. It holds that $r\wal s$ iff both are $1$
  or $r < s$.\footnote{Note that $1\le \bigsqcap F$ can only hold if
    $F = \{1\}$ due to the requirement that every filtered set is
    non-empty.}  The lattice $[0,1]$ is co-continuous, since every
  number is the infimum of numbers that are strictly larger.
\end{example}

We need the notion of \emph{basis}, i.e., a subset of the lattice that
allows generating each element (either via join or meet).  A
\emph{join basis} of a lattice $\lt{L}$ is a subset
$\jb{L}\subseteq \lt{L}$ such that for each element $\ell\in \lt{L}$,
$\ell = \bigsqcup \{b\in \jb{L} \mid b\sqsubseteq \ell\}$.  A
\emph{meet basis} of $\lt{L}$ is a subset $\mb{L}\subseteq \lt{L}$
such that for each element $\ell\in \lt{L}$,
$\ell = \bigsqcap \{m\in \mb{L} \mid m\sqsupseteq \ell\}$.

When $\lt{L}$ is a continuous lattice and $\jb{L}$ is a join basis
of $\lt{L}$, then it holds -- for all $\ell\in\lt{L}$ -- that
$\ell = \bigsqcup_{b\ll \ell, b\in \jb{L}} b$. 
The dual holds for a meet basis and the way-above relation.

\begin{example}
  For a powerset lattice $(\mathcal{P}(X),\subseteq)$, one choice of
  join basis is to consider all singletons. A possible meet basis is
  the set of all co-singletons, i.e., all sets of the form
  $\overline{\{x\}} = X\backslash\{x\}$ for $x\in X$.

  For the lattice $([0,1],\le)$ a possible join basis consists of
  $(0,1]\cap\mathbb{Q}$.
\end{example}

We also need the following notion of irreducibility. We call an element 
$b\in \lt{L}$ of a lattice $\lt{L}$ \emph{way-below irreducible} 
(or simply \emph{irreducible})
if whenever $b\ll \bigsqcup F$ for $F\subseteq \lt{L}$ finite, then
$b\ll f$ for some $f\in F$ (where $F$ is not necessarily directed).

Note that since $\bot \ll \bot = \bigsqcup \emptyset$, $\bot$ can
  not be an irreducible. 

\begin{example}
  For the powerset lattice $(\mathcal{P}(X),\subseteq)$, singletons
  are way-below irreducible. For the lattice $([0,1],\le)$, all
  elements -- apart from $0$ -- are way-below irreducible.
\end{example}

\subparagraph*{Galois Connections and Adjoint Logic:}

We summarize the Galois connection approach from \cite{HMLGalois} that
can be used to give an abstract account of the Hennessy-Milner
theorem.  Intuitively, the Galois connection relates a ``logical
universe'' and a ``behavioural universe''.

Let $\lt{L}$, $\lt{B}$ be two lattices.  A \emph{Galois
  connection} from $\lt{L}$ to $\lt{B}$ is a pair
$\alpha\dashv \gamma$ of monotone functions
$\alpha\colon \lt{L}\to\lt{B}$,
$\gamma\colon \lt{B}\to\lt{L}$ such that for all $\ell\in \lt{L}$:
$\ell\sqsubseteq \gamma(\alpha(\ell))$ and for all $d\in \lt{B}$:
$\alpha(\gamma(d)) \sqsubseteq d$.
It is well-known that a left adjoint $\alpha$ preserves all suprema
while a right adjoint $\gamma$ preserves all infima.

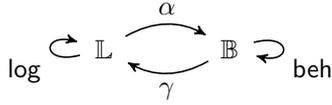
\begin{wrapfigure}{l}{.32\textwidth}
    \begin{tikzpicture}[node distance=1.5 and 1, shorten >=1pt, >=stealth', semithick]
      \begin{scope}[state, inner sep=3pt, minimum size=0pt]
        \draw node [] (q1) {\( \lt{L} \)};
        \draw node [right=of q1] (q2) {\(\lt{B}\)};
      \end{scope}
  
      \begin{scope}[->]
        \draw (q1) edge [bend left] node [midway, above] {\(\alpha\)} (q2);
        \draw (q2) edge [bend left] node [midway, below] {\(\gamma\)} (q1);
        \draw (q1) edge [loop left] node [below left] {\(\lo\)} (q1);
        \draw (q2) edge [loop right] node [below right] {\(\be\)} (q2);
      \end{scope}
    \end{tikzpicture}
\end{wrapfigure}

\noindent We also fix two monotone functions
$\lo\colon\lt{L}\to \lt{L}$ and $\be\colon \lt{B}\to \lt{B}$ and by
Knaster-Tarski \cite{t:lattice-fixed-point} they both have least
fixpoints (denoted by $\mu\lo$, $\mu\be$).
  
From \cite{HMLGalois} it follows that $\alpha$ preserves least
fixpoints of $\lo$, $\be$ ($\alpha(\mu \lo) = \mu \be$) whenever
$\alpha\circ\lo = \be\circ \operatorname{\alpha}$.  This condition
also implies that $\alpha$ preserves all stages of the Kleene
iteration over the ordinals, i.e., $\alpha(\lo^i(\bot)) = \be^i(\bot)$
for all $i\in \mathsf{Ord}$.

\begin{example}[Bisimilarity \& Hennessy-Milner Logic]
  \label{ex:running-intro}
  As a running example we consider standard bisimilarity
  \cite{s:bisimulation-coinduction}.
 Here we follow
  \cite{HMLGalois} (using arbitrary relations instead of only
  equivalence relations) and work -- for simplicity -- with unlabelled
  transition systems, i.e., pairs $(X, \to)$ consisting of a state
  space $X$ and a transition relation $\to\ \subseteq\,X \times X$. We
  write 
  $\mathit{succ}(x) = \{x'\in X\mid x\to x'\}$ and assume finite
  branching, i.e., $\mathit{succ}(x)$ is finite for all $x\in X$. 
  
We define lattices
$\lt{L} = (\mathcal{P} (\mathcal{P}(X)),\subseteq)$,
$\lt{B} = (\mathit{Rel}(X),\supseteq)$, where $\mathit{Rel}(X)$ is the
set of all relations $R\subseteq X\times X$. The Galois connection is:
  \begin{eqnarray*}
    \alpha(\mathcal{S}) & = & \{(x_1,x_2)\in X\times X \mid \forall S\in
    \mathcal{S}\colon (x_1\in S\iff x_2\in S)\} \\
    \gamma(R) & = & \{ S\subseteq X \mid \forall (x_1,x_2)\in R\colon
    (x_1\in S\iff x_2\in S) \}. 
  \end{eqnarray*}
  Intuitively, $\alpha$ generates an equivalence on $X$ from a set of
  subsets of $X$ and $\gamma$ maps a relation to all subsets of $X$
  that are closed under this relation.

  As logic function we consider
  $\lo\colon \mathcal{P}(\mathcal{P}(X))\to
  \mathcal{P}(\mathcal{P}(X))$ with
  $\lo(\mathcal{S}) = \Diamond[\cl^{\cup,\lnot}_f(\mathcal{S})]$,
  where $f[A] = \{f(a)\mid a\in A\}$ for a function $f\colon X\to Y$ and
    $A\subseteq X$, $\cl^{\cup,\lnot}_f$ closes $\mathcal{S}$ under finite
  intersections and complement and
  $\Diamond(S) = \{x\in X \mid \exists x'\in S:\ x \to x'\}$.

  The corresponding behaviour function
  $\be\colon \mathit{Rel}(X)\to \mathit{Rel}(X)$ is the standard
  (monotone) bisimilarity map: given $R\in \mathit{Rel}(X)$ it
  holds that $x_1 \mathrel{\be (R)} x_2$ iff
  \[
    \forall y_1\in\mathit{succ}(x_1)\,\exists {y_2\in\mathit{succ}(x_2)}
    \colon y_1 \mathrel R y_2 \land \forall
    y_2\in\mathit{succ}(x_2)\,\exists {y_1\in\mathit{succ}(x_1)} \colon
    y_1 \mathrel R y_2.
  \]
  From \cite{HMLGalois} it follows that
  $\alpha\circ \lo = \be\circ\operatorname{\alpha}$.  In particular
  this means that $(x_1,x_2)\in \alpha(\mu\lo)$ iff $x_1,x_2$ are
  bisimilar (denoted by $x_1\sim x_2$), which is exactly the statement
  of the Hennessy-Milner theorem \cite{hm:hm-logic}.
\end{example}

In the rest of the paper, we assume that $\lt{L}$ and $\lt{B}$
are both complete lattices, with monotone endofunctions $\lo$
respectively $\be$, and a Galois connection $\alpha\dashv \gamma$ as
above, such that $\alpha\circ\lo = \be\circ \operatorname{\alpha}$.
Intuitively, this means that the functions $\lo$ and $\be$ satisfy
the natural requirement of being ``in step'' with each other.

\subparagraph*{Scott Topology:} 

We need some concepts from Scott topology, in particular open sets, as
well as continuous and proper maps (cf. \cite{CompContLattices}).

A subset $O\subseteq \lt{L}$ of a complete lattice
$(\lt{L},\sqsubseteq)$ is called \emph{(Scott) open} iff
$O= \mathop{\uparrow} O$, and $\bigsqcup D\in O$ implies
$D\cap O \not= \emptyset$ for all directed sets $D\subseteq\lt{L}$.
The collection of all Scott open subsets of $\lt{L}$ is called the
\emph{Scott topology} of $\lt{L}$ and denoted by $\sigma(\lt{L})$.  A
set $C\subseteq \lt{L}$ is called \emph{compact} if whenever $C$ can
be covered by open sets ($C\subseteq \bigcup_{i\in I} O_i$, $O_i$
open), there exists a finite subcover
($C\subseteq \bigcup_{i\in J} O_i$ where $J\subseteq I$ and $J$
finite).

In a continuous lattice the sets of the form $\mathop{\twoheaduparrow} \ell$
($\ell\in \lt{L}$) provide a basis for the Scott open sets
\cite{CompContLattices}. That is, each Scott open set is the union of
such sets.

Furthermore each set of the form $\mathop{\uparrow} \ell$ (for
$\ell\in \lt{L}$) and their finite unions can easily be seen to be
compact.

\begin{example}
  \label{ex:powerset-open-compact}
  In a powerset lattice $\mathcal{P}(X)$, the open sets are the
  families of finite character, i.e., $O\subseteq \mathcal{P}(X)$ such
  that $S\in O$ iff $F\in O$ for some finite subset $F$ of $S$ 
  \cite{CompContLattices}. 
\end{example}

A function $f:\lt{S}\rightarrow \lt{T}$ between lattices is called
\emph{(Scott-)continuous} if and only if it is continuous with respect
to the Scott topologies (i.e. $f^{-1}(U)\in\sigma(\lt{S})$ for all
$U\in\sigma(\lt{T})$).  Equivalently, $f$ is continuous if it preserves
suprema of directed sets (i.e. $f(\bigsqcup D)= \bigsqcup f(D)$ for
all directed subsets $D$ of $\lt{S}$).  
Note that a continuous function reaches its smallest fixpoint in
$\omega$ steps.

A map
$f:\lt{S}\rightarrow \lt{T}$ is \emph{proper} if the inverse image of
a compact set is also compact. It is called proper wrt.
$B\subseteq \lt{T}$ if the inverse image of each $\mathop{\uparrow} \ell$ (for
$\ell\in B$) is compact.

All these notions have duals that can be obtained by flipping the
order and will be denoted by co-open, co-compact, co-continuous and
co-proper.

\begin{examplerep}[Continuity and Properness]
  \label{ex:running-continuous-proper}
  Consider the bisimilarity map $\be$ introduced earlier
  (cf. Example~\ref{ex:running-intro}). Since the transition system is
  assumed to be finitely branching, this map is co-continuous wrt.\
  $\subseteq$ (see also \cite{s:bisimulation-coinduction}).  It is
  also proper wrt.\ $\subseteq$ and a basis consisting of singletons
  $\{(x_1,x_2)\}\subseteq X\times X$.  Wrt.\ $\supseteq$ the map is
  hence continuous and co-proper wrt.\ a basis of co-singletons.
\end{examplerep}

\begin{proof}
  
  We first consider the co-continuity of $\be$: for $x_1,x_2\in X$ we
  have to show that $(x_1,x_2)\in \bigcap_{i\in I} \be(R_i)$ implies
  $(x_1,x_2)\in \be(\bigcap_{i\in I} R_i)$ (for a filter
  $(R_i)_{i\in I}$), the other direction is obvious. Assume that
  $(x_1,x_2)\in \bigcap_{i\in I} \be(R_i)$ and $x_1\to y_1$. Then
  $x_2$ has an answering move for each $i\in I$, i.e., $x_2\to y_2^i$
  such that $(y_1,y_2^i)\in R_i$. Since we have finite branching the
  set $F = \{(y_1,y_2^i)\mid i\in I\}$ is finite. Hence some
  $(y_1,y_2^i)$ must be contained in the intersection
  $\bigcap_{i\in I} R_i$. If that were not the case, then for each
  element of $p\in F$ there exists $R_{i_p}$ that does not contain
  $p$. Then a set $S\subseteq\bigcap_{p\in F} R_{i_p}$ is in the
  filter, but does not contain any element of $F$, which is a
  contradiction. (Analogously for a move of $x_2$.) Since for all
  $y_1$ with $x_1\to y_1$ some $(y_1,y_2^i)$ is contained in
  $\bigcap_{i\in I} R_i$ (and vice versa), we can infer that
  $(x_1,x_2)\in\be(\bigcap_{i\in I} R_i)$.

  Furthermore $\be$ is proper: given a basis element
  $b = \{(x_1,x_2)\}$, it holds that
  $\be^{-1}(\mathop{\uparrow} b) = \{C \subseteq X\times X\mid
  (x_1,x_2)\in \be(C)\}$. Since the transition system is finitely
  branching ($\mathit{succ}(x_1)$, $\mathit{succ}(x_2)$ are finite),
  there are only finitely many minimal such couplings $C$ of
  $\mathit{succ}(x_1),\mathit{succ}(x_2)$ and the set is hence
  compact.
\end{proof}

\section{Witnesses}
\label{sec:dist-form}

We will now introduce witnesses, analogous to distinguishing formulas
in Hennessy-Milner logic that witness the non-bisimilarity or distance
between two states (see \cite{stirling99}). A witness guarantees
a lower bound for the least fixpoint respectively shows that a given
element is not an upper bound.

We will later explain how witnesses can be constructed and how they
are related to the winning strategies of the attacker in a game.
We distinguish between primal and dual witnesses depending on the type
of game for which they are used.
In order to clarify the notation, we decorate lattice elements with a
dot in the dual case.

\begin{definition}[witness]
  \label{def:witness}
  Let $\lt{L}$, $\lt{B}$ be two lattices with a Galois
  connection $\alpha\dashv \gamma$ and monotone functions
  $\lo\colon\lt{L}\to \lt{L}$,
  $\be\colon\lt{B}\to \lt{B}$. Let $b\in\lt{B}$.

  A \emph{primal witness} for $b$ is an element
  $\pr{a}\in\lt{L}$ such that $\pr{a}\ll \mu\lo$ and
  $b\ll \alpha(\pr{a})$.

  A \emph{dual witness} for $\du{b}$ is an element
  $a\in\lt{L}$ such that $a\ll \mu\lo$ and
  $\alpha(a) \not\sqsubseteq \du{b}$.
\end{definition}

Intuitively a witness is a formula of the logic that is strong enough
to show that either $b$ is a lower bound or $\du{b}$ is not an 
upper bound for the least fixpoint of $\be$.

\begin{toappendix}
\begin{lemma}
  \label{lem:auxiliary-functions-1}
  Let $\lt{L}$ be a continuous lattice with join basis $\jb{L}$ and
  the join basis $\jb{B}$ of the lattice $\lt{B}$ contain only
  irreducibles.
  \begin{enumerate}
  \item \label{lem:aux-1-1} Let $b\in \jb{B}$ and $\ell \in \lt{L}$
    with $b\ll \alpha(\ell)$. Then there exists $a\in \jb{L}$ 
    with $a\ll \ell$ and $b\ll \alpha(a)$.  
  \end{enumerate}
  Let $\lt{L}$ be a a continuous lattice with join basis $\jb{L}$.
  \begin{enumerate}
    \setcounter{enumi}{1}
  \item \label{lem:aux-1-2} Let $\du{b}\in \lt{B}$ and
    $\ell\in\lt{L}$ with $\alpha(\ell) \not\sqsubseteq \du{b}$. Then
    there exists $a\in \jb{L}$ 
    with $a\ll \ell$ and
    $\alpha(a)\not\sqsubseteq \du{b}$.
  \end{enumerate}
\end{lemma}

\begin{proof}  
  \begin{enumerate}
  \item Since $\lt{L}$ is continuous and  $\alpha$ is join-preserving
    it holds that
    \[
      b\ll \alpha(\ell) =
      \alpha(\bigsqcup_{\substack{a\in \jb{L}\\
          a\ll \ell}} a) =
      \bigsqcup_{\substack{a\in \jb{L}\\
          a\ll \ell}} \alpha(a)
    \]
    Since $b$ is a way-below irreducible, there exists
    $a\in \jb{B}$ with $a\ll \ell$ and $b\ll \alpha(a)$.
  \item Due to the fact that $\lt{L}$ is continuous and $\alpha$ is
    join-preserving we have
    \[ \bigsqcup \{\alpha(a)\in \jb{L}\mid a\ll \ell\} =
      \alpha(\bigsqcup \{a\in \jb{L}\mid a\ll \ell\}) =
      \alpha(\ell) \not\sqsubseteq \du{b}. \] Hence there must be
    $a\ll \ell$ such that $\alpha(a) \not\sqsubseteq \du{b}$. \qedhere
  \end{enumerate}
\end{proof}
\end{toappendix}

\begin{propositionrep}
  Assume the setting of Definition~\ref{def:witness}. Let $\lt{L}$ be
  continuous with join basis~$\jb{L}$.
  
  Assume that $\lt{B}$ has a join basis $\jb{B}$ consisting only of
  way-below irreducibles. Given $b\in \jb{B}$, there exists a primal
  witness $\pr{a}\in \jb{L}$ for $b$ iff $b\ll \mu\be$.

  Given $\du{b}\in\lt{B}$, there exists a dual witness $a\in \jb{L}$ for
  $\du{b}$ iff $\mu\be\not\sqsubseteq \du{b}$. 
\end{propositionrep}
  
\begin{proof}
    Let $\pr{a}$ be a primal witness for $b$. Hence
    $b\ll \alpha(\pr{a}) \sqsubseteq \alpha(\mu\lo) = \mu\be$. Hence
    clearly $b\ll \mu\be$.
  
    Now assume that $b\ll\mu\be = \alpha(\mu\lo)$ where
    $b\in \jb{B}$. Then the statement follows from
    Lemma~\ref{lem:auxiliary-functions-1}(\ref{lem:aux-1-1}).
    
    \smallskip\hrule\smallskip
    
    Let $a$ be a dual witness for $\du{b}$, that is
    $a\sqsubseteq \mu\lo$ and $\alpha(a)\not\sqsubseteq \du{b}$. Since
    $a\sqsubseteq \mu\lo$, by monotonicity we have that
    $\alpha(a)\sqsubseteq \alpha(\mu \lo) = \mu \be$. Assume that
    $\mu \be\sqsubseteq \du{b}$. This implies $\alpha(a)\sqsubseteq \du{b}$, a
    contradiction. 
  
    Now assume that $\mu \be \not\sqsubseteq \du{b}$. Furthermore $\mu\be =
    \alpha(\mu\lo)$. Then the statement follows from
    Lemma~\ref{lem:auxiliary-functions-1}(\ref{lem:aux-1-2}).
    
\end{proof}

Note that for primal witnesses we have to assume a basis of
  way-below irreducibles, which is a restriction. Because of the
above proposition, we can focus our attention on the chosen
(join or meet) basis and will typically assume that the witness
$a\in\lt{L}$ is a basis element. The same holds for
$b\in\lt{B}$.

Existence of witnesses is rather easy to prove, however the
construction of such witnesses by a recursive process is more
challenging and will be treated in the next sections.
  
\begin{example}[Bases \& Witnesses]
  \label{ex:running-basis-witness}
  We continue our running example and instantiate these concepts to
  the case of unlabelled transition systems and bisimilarity.

  For a join basis of $\lt{L} = \mathcal{P}(\mathcal{P}(X))$ we
  consider singleton sets $\{Y\}$ where $Y\subseteq X$.

  For $\lt{B} = \mathcal{P}(X\times X)$ (ordered by $\subseteq$) we
  consider a join basis consisting of all singletons
  $\{(x_1,x_2)\}\subseteq X\times X$ (which are also irreducibles). As
  a meet basis we use all co-singletons
  $\overline{\{(x_1,x_2)\}} := (X\times X)\backslash \{(x_1,x_2)\}$
  for $x_1,x_2\in X$. 

  If we reverse the order to $\supseteq$, the join basis becomes a
  meet basis and vice versa and all join basis elements are
  irreducible. We will now continue working wrt.\ the inverse
  inclusion.

  It holds that
  $\overline{\{(x_1,x_2)\}} \ll \mu\be \iff \mu \be \subseteq
  \overline{\{(x_1,x_2)\}} \iff (x_1,x_2)\not\in\mu\be \iff
  x_1\not\sim x_2$.
  A primal witness is then a predicate $\pr{a} =
  \{X'\}\subseteq\mu\lo$ (with $X'\subseteq X$)
  that is obtained by evaluating a logical formula, for which
  $\overline{\{(x_1,x_2)\}} \ll \alpha(\{Y\})\iff
  \alpha(\{Y\})\subseteq \overline{\{(x_1,x_2)\}} \iff
  \lnot(x_1\,\alpha(\{Y\})\,x_2) \iff (x_1\in Y\Leftrightarrow
  x_2\not\in Y)$. And that is exactly the notion of a distinguishing
  formula.

  In the dual case
  $\mu\be\not\sqsubseteq \{(x_1,x_2)\} \iff \mu\be\not\supseteq
  \{(x_1,x_2)\} \iff (x_1,x_2)\not\in \mu\be\iff x_1\not\sim x_2$. As
  in the primal case a witness is a predicate obtained from a logical
  formula that distinguishes $x_1,x_2$. 
\end{example}

We will now define the notion of (co-)degree, which is a ordinal that
gives the least number of iterations needed to cover (or reach) a
certain element.

\begin{definition}
  \label{def:co-degree}
  Let $\lo\colon\lt{L}\to\lt{L}$.  The \emph{degree} and
  \emph{co-degree} of an element $a\in\lt{L}$ (wrt.\ $\lo$) are
  defined as follows (where we assume that $\min\emptyset$ is
  undefined):
  \[ \degr_{\lo}(a) = \min\{i\in\mathsf{Ord}\mid a \ll 
    \lo^{i}(\bot) \} \qquad \cdegr_{\lo}(a) =
    \min\{i\in\mathsf{Ord}\mid \lo ^{i}(\bot) \not\sqsubseteq a\} \]
\end{definition}

Note that the degree (co-degree) of
an element $a$ is defined whenever $a\ll \mu\lo$
($\mu\lo\not\sqsubseteq a$). It holds that $\degr(\bot) = 0$,
otherwise the (co-)degree is always a successor ordinal. We will
often omit the subscript and simply write $\degr$ or $\cdegr$ if the
function ($\lo$ or $\be$) is clear from the context.

\begin{toappendix}
  \begin{lemma}
    \label{lem:properties-deg}
    It holds that
    \begin{enumerate}
    \item \label{lem:properties-deg-1} $\degr$ is monotone, i.e.,
      $a\sqsubseteq a'$ implies $\degr(a)\le \degr(a')$.
    \item \label{lem:properties-cdeg} $\cdegr$ is monotone, i.e.,
      $a\sqsubseteq a'$ implies $\cdegr(a)\le \cdegr(a')$.
    \item \label{lem:properties-deg-2} if $a'\ll \lo(a)$, then
      $\degr(a')\leq \degr(a)+1$.
    \item \label{lem:properties-deg-3}
      $\degr(a\sqcup a') = \max\{\degr(a),\degr(a')\}$.
    \end{enumerate}
  \end{lemma}

\begin{proof}
  \mbox{} Note that $x\sqsubseteq y\ll z$ implies $x\ll z$; and that
  $x\ll z$ and $y\ll z$ implies $x\sqcup y\ll z$
  \cite{CompContLattices}.
  \begin{enumerate}
  \item Follows directly from the definition: since
    $a\sqsubseteq a'\ll \lo^{\degr(a')}(\bot)$,
    $a\ll \lo^{\degr(a')}(\bot)$ and therefore
    $\degr(a)\leq \degr(a')$.
  \item Let $a\sqsubseteq a'$ and let $\cdegr(a) = k$ and
    $\cdegr(a') = k'$.  This implies that
    $\lo^{k'}(\bot) \not\sqsubseteq a'$. Hence
    $\lo^{k'}(\bot) \not\sqsubseteq a$, since otherwise
    $\lo^{k'}(\bot) \sqsubseteq a \sqsubseteq a'$, a
    contradiction. Hence $\cdegr(a) \le k' = \cdegr(a')$.
  \item Let $\degr(a) = k$ and we have $a\ll \lo^k(\bot)$ and
    therefore $a\sqsubseteq\lo^k(\bot)$.  Hence by monotonicity of
    $\lo$ it holds that $\lo(a)\sqsubseteq \lo^{k+1}(\bot)$.  Since
    $a'\ll\lo(a)$, we get
    $a'\ll\lo(a)\sqsubseteq\lo^{k+1}(\bot)$. Hence,
    $a\ll\lo^{k+1}(\bot)$, which means $\degr(a') \le k+1=\degr(a)+1$.
  \item We set $k=\degr(a)$, $k'=\degr(a')$ and assume that $k\le k'$.
    We show both inequalities:
    \begin{itemize}
    \item $\degr(a\sqcup a') \le k'$: From
      $a\ll \lo^k(\bot)\sqsubseteq \lo^{k'}(\bot)$, we get
      $a\ll \lo^{k'}(\bot)$. Together with $a'\ll \lo^{k'}(\bot)$,
      this implies $a\sqcup a' \ll \lo^{k'}(\bot)$ and therefore
      $\degr(a\sqcup a')\leq k'$.
    \item $\degr(a\sqcup a') \ge k'$: from $a\sqsubseteq a\sqcup a' $
      and with (\ref{lem:properties-deg-1}), it follows
      $\degr(a')=k'\leq \degr(a\sqcup a')$. \qedhere
    \end{itemize}
  \end{enumerate}
\end{proof}

\begin{lemma}
  \label{lem:witness-degr}
  Let $\lo\colon\lt{L}\to\lt{L}$, $\be\colon \lt{B}\to \lt{B}$ with
  $\alpha\circ\lo = \be\circ\operatorname{\alpha}$. Let
  $\pr{a}\in\lt{L}$ be a primal witness for $b\in \lt{B}$. Then
  $ \degr(b)\le \degr(\pr{a})$.
    
  Analogously, let $a\in \jb{L}$ be a dual witness for
  $b\in \lt{B}$, i.e., $\alpha(a)\not\sqsubseteq b$.  Then
  $\cdegr(b)\le \degr(a)$.
\end{lemma}
  
\begin{proof}
  First we assume $\pr{a}$ is a primal witness for $b$ and prove
  $\degr(b)\leq\degr(\pr{a})$.  We have $b\ll \alpha(\pr{a})$ and
  $\pr{a}\ll \mu\lo$. Let $i=\degr(\pr{a})$, i.e.,
  $a\ll \lo^i(\bot)$. This implies
  \[ b\ll\alpha(\pr{a}) \sqsubseteq \alpha(\lo^i(\bot)) =
    \be^i(\alpha(\bot)) = \be^i(\bot). \] Hence
  $\degr(b)\le i = \degr(\pr{a})$.
  
  Regarding the dual witness, let $\degr(a)=i$. By definition,
  $a\sqsubseteq \lo^i(\bot)$.  Applying $\alpha$ yields
  $\alpha(a)\sqsubseteq \alpha(\lo^i(\bot)) = \be^i(\bot)$.
  Since $a$ is a dual witness for $b$, $\alpha(a)\not\sqsubseteq b$.
  From $\alpha(a)\sqsubseteq \be^i(\bot)$ we can infer that
  $\be^i(\bot)\not\sqsubseteq b$.  Therefore,
  $\min\{i\in \mathsf{Ord}\mid \be^{i}(\bot) \not\sqsubseteq b\} \leq
  i$ and hence $\cdegr(b) \leq \degr(a)$.
\end{proof}
\end{toappendix}

\section{Games}

In this section we will assume a monotone function on a complete
lattice and study fixpoint games \cite{POPL19FPGames} that
characterize the least fixpoint of such functions. Our aim is in
particular to single out cases where finite strategies are sufficient
and eventually connect game strategies with witnesses.  In the
following we spell out the games for a function
$\be\colon \lt{B}\to\lt{B}$, but we will also play the primal way
below game for the ``other'' function $\lo\colon \lt{L}\to \lt{L}$.

We first introduce the game as it is presented in
\cite{POPL19FPGames}, in the following also called the primal game.
Let $\lt{B}$ be a continuous lattice with a join basis $\jb{B}$ such
that $\bot\not\in \jb{B}$. Let $\be\colon \lt{B}\to \lt{B}$ be
monotone. The game starts with $b\in\jb{B}$ and players $\exists$ and
$\forall$ play according to the following rules:

\begin{mdframed}
  \it
  \vspace{-0.2cm}
  \begin{center}
    \begin{tabular}{ c c l }
      \normalsize
      Position & \normalsize Player & \normalsize Moves \\
      \hline
      \normalsize $b\in \jb{B}$ & \normalsize $\exists$ & \normalsize $d\in\lt{B}$, such that
      $b \sqsubseteq \be(d)$ \\
      \normalsize $d\in \lt{B}$ & \normalsize $\forall$ & \normalsize $b'\in \jb{B}$ such that
      $b' \ll d$
    \end{tabular}
  \end{center}
  \vspace{-0.2cm}
\end{mdframed}

If a player cannot move, the opponent wins. Infinite games are won by
$\forall$.

It holds that $b\sqsubseteq \mu\be$ iff the $\exists$ player has a
winning strategy from $b$ \cite{POPL19FPGames}. The intuitive
  reason for requiring the way-below order for the answer of $\forall$
  is that -- whenever $b\sqsubseteq \mu\be$ -- $\exists$ can ensure
  that the game positions descend the chain of ordinals and $\forall$
  eventually runs out of answering moves.
  
Furthermore $b\sqsubseteq \nu\be$ iff $\exists$ has a strategy in the
same game with the only modification that $\exists$ wins infinite
games.

\subsection{Primal Way-Below Game}
\label{sec:primal-way-below-game}
  
There is also a ``way-below'' version of the primal game where the aim
of $\exists$ is to show a strict lower bound, i.e.\ prove that
$b\ll \mu\be$. To our knowledge this version of the game is original.

\begin{definition}
    \label{def:primal-way-below-game}
    Let $\lt{B}$ be a continuous lattice with a join basis
    $\jb{B}$ such that $\bot\not\in \jb{B}$. Let
    $\be\colon \lt{B}\to \lt{B}$ be monotone. We define a game
    between players $\exists$ and $\forall$ that play according to the
    following rules starting with $b\in\jb{B}$:
  
    \begin{mdframed}
      \vspace{-0.2cm}
      \begin{center}
        \begin{tabular}{ c c l }
          \normalsize
          Position & \normalsize Player & \normalsize Moves \\
          \hline
          \normalsize $b\in \jb{B}$ & \normalsize $\exists$ &
          \normalsize $d\in\lt{B}$, such that
          $b \ll \be(d)$ \\
          \normalsize $d\in \lt{B}$ & \normalsize $\forall$ &
          \normalsize $b'\in \jb{B}$ such that
          $b' \ll d$
        \end{tabular}
      \end{center}
      \vspace{-0.2cm}
    \end{mdframed}
  
    If a player cannot move, the opponent wins. Infinite games are won
    by $\forall$.
  \end{definition}

In the context of the primal way-below game, we will also call the
$\exists$ player the attacker and the $\forall$ player the defender. 

Note that since in parity games memoryless winning strategies are
always enough \cite{bw:mu-calculus-modcheck}, it suffices to consider
positional winning strategies. We can also show that under some
circumstances formulas are ``finitely constructed'', i.e., strategies
are finitary.

\begin{propositionrep}
    \label{prop:winning-strategy-primal-way-below}
    Let $b\in \jb{B}$. Player $\exists$ has a winning strategy in
    the game in Definition~\ref{def:primal-way-below-game} iff
    $b\ll \mu\be$. 

    Whenever $\be$ is a continuous function, then $\exists$ has a
    finitary winning strategy. Finitary means that $d = \bigsqcup F$
    where $F$ is a finite subset of $\jb{B}$ and
    $\degr(b') < \degr(b)$ for each $b'\in F$. In particular $\exists$
    can win in $\degr(b)$ steps.
\end{propositionrep}
  
\begin{proof}
  Assume that $b\ll\mu\be = \bigsqcup_{i\in\mathsf{Ord}}
  \be^i(\bot)$. The latter is a directed join, hence,
  $b\ll \be^i(\bot)$ for some $i\in \mathsf{Ord}$ (cf. \cite[Thm
  I-1.9]{CompContLattices}).  We take the least such ordinal
  (i.e. $i=\degr(b)$) and observe that it must be a successor ordinal
  $i=j+1$. Now $\exists$ plays
  $\strat_{p,\be}^\exists(b) = \be^j(\bot)$, which means that
  $\forall$ must choose $b'\ll \be^j(\bot)$.
  Hence $\degr(b')\leq j < i = \degr(b)$.  The game continues and
  terminates after finitely many steps when the degree reaches $0$ and
  $\forall$ must choose $b'\in \jb{B}$ with
  $b'\ll \be^0(\bot) = \bot$. This is not possible since
  $\bot\notin \jb{B}$.

  Since every round is played from a position of lower degree than the
  previous round, and the game ends for a position of $\degr$ zero,
  the player $\exists$ wins the game from position $b$. It takes at
  most $\degr(b)$ steps if the ordinal $\degr(b)$ is a natural number.

  \smallskip\hrule\smallskip
  
  Now assume that $b\not\ll\mu\be$. We show that there is a winning
  strategy for $\forall$. Let $\exists$ play some $d$ with
  $b\ll\be(d)$. Then there exists $b'\in \jb{B}$, $b'\not\ll\mu\be$
  with $b'\ll d$. If that were not the case, it holds that all
  $b'\ll d$ satisfy $b'\ll\mu \be$, hence by the fact that $\lt{B}$ is
  continuous
  \[ d = \bigsqcup_{b'\ll d, b'\in \jb{B}} b' \sqsubseteq
    \bigsqcup_{b'\ll \mu\be, b'\in \jb{B}} b' = \mu\be \] and this
  implies by monotonicity
  \[ b\ll \be(d) \sqsubseteq \be(\mu \be) = \mu \be, \] a
  contradiction. The player $\forall$ plays this $b'$ and can keep the
  game going forever.

  \smallskip\hrule\smallskip
  
  We will argue that it is enough if only finitary strategies are
  played, i.e., $d$ is of the form $\bigsqcup F$ where $F$ is finite
  subset of $\jb{B}$ and $\degr(b') < \degr(b)$ for each $b'\in
  F$. 

  We show that whenever $d = \strat_{p,\be}^\exists(b)$
  (i.e., $b\ll \be(d)$), there exists $F\subseteq \jb{B}$
  finite, such that $\bigsqcup F\sqsubseteq d$ and
  $b\ll \be(\bigsqcup F)$, i.e., playing $\bigsqcup F$ is also a
  winning strategy.

  It holds that $d\in \be^{-1}(\mathop{\twoheaduparrow} b)$. We represent $d$
  as a directed join of finite joins, that is
  \[ d = \bigsqcup _{\substack{F\subseteq \{b'\mid b'\ll d\}\\
        F\subseteq \jb{B}\text{\scriptsize\ finite}}} \bigsqcup F \]
  Since $\be$ is continuous, the set
  $\be^{-1}(\mathop{\twoheaduparrow} b)$ is open.  Hence, there must
  be a finite set $F$ such that $\bigsqcup F$ is contained in
  $\be^{-1}(\mathop{\twoheaduparrow} b)$ (hence $b\ll \bigsqcup F$)
  and below $d$ ($\bigsqcup F \sqsubseteq d$). If playing $d$ is a
  valid strategy, playing $\bigsqcup F$ is also a valid strategy.
\end{proof}

Since it is guaranteed that the degrees decrease the game will
terminate eventually and $\exists$ wins. Such a finitary winning
strategy for $\exists$ (the attacker) is denoted by
$\strat_{p,\be}^\exists$ and assigns a suitable move
$F\subseteq \jb{B}$ to $b\in \jb{B}$. Then $\exists$ plays
$\bigsqcup \strat_{p,\be}^\exists(b)$. $\strat_{p,\be}^\exists(b)$ is
undefined if there is no winning strategy from $b$.

\begin{example}[Strategy for the Primal Game]
  \label{ex:running-primal-way-below-game}
  We spell out the primal game for our running example. The join basis
  for $(\mathcal{P}(X\times X),\supseteq)$ are the co-singletons
  (cf. Example~\ref{ex:running-basis-witness}). Now, when the initial
  situation is $\overline{\{(x_1,x_2)\}}$, the game proceeds as
  follows:
  \begin{itemize}
  \item $\exists$ (attacker) plays $R\subseteq X\times X$ such that
    $\be(R)\subseteq \overline{\{(x_1,x_2)\}}$ (i.e.,
    $(x_1,x_2) \not\in \be(R)$).

    More concretely $\exists$ plays $R$ satisfying: there exists
    $y_1\in X$ with $x_1\to y_1$ such that for all $y_2\in X$ with
    $x_2\to y_2$ it holds that $(y_1,y_2)\not\in R$ \emph{or} vice
    versa.
  \item $\forall$ (defender) chooses
    $R\subseteq \overline{\{(y_1,y_2)\}}$ (i.e., $(y_1,y_2)\not\in R$)
    and wins infinite games. 
  \end{itemize}
  If $x_1\not\sim x_2$, a possible winning strategy of $\exists$
  (attacker) is to play
  \[ \strat_{p,\be}^\exists(\overline{\{(x_1,x_2)\}}) =
    \overline{\{y_1\}\times \mathit{succ}(x_2)} =
    \bigcap\nolimits_{y\in\mathit{succ}(x_2)}
    \overline{\{(y_1,y)\}} \] if $x_1\stackrel{a}{\to} y_1$ is a
  winning move in the traditional game (or vice versa). By deriving
  the strategy from the fixpoint iteration it is always possible to
  choose $y_1$ such that for all $y\in \mathit{succ}(x_2)$,
  $\overline{\{(y_1,y)\}}$ has a smaller degree than
  $\overline{\{(x_1,x_2)\}}$. Then $\forall$ must pick $(y_1,y)$,
  which is equivalent to making an answering move in the traditional
  game.  Hence we obtain a game that is close to the game in
  \cite{stirling99} where a defender mimics the moves made by the
  attacker.
\end{example}

In Section~\ref{sec:dist-form} we considered (primal) witnesses as
monolithic objects. But in fact, in order to use witnesses to derive
strategies, we have to be able to deconstruct them into subformulas.
This is exactly what this game does for the logic function $\lo$ (as
defined in Example~\ref{ex:running-intro}). Given $a\in \jb{L}$ with
$a\ll \mu \lo$ (a basis element, representing to a formula of the
logic), $\exists$ is obliged to play a finite $A\subseteq \jb{L}$ with
$a\sqsubseteq \lo(\bigsqcup A)$. That is, $\exists$ exhibits the
subformulas, from which $a$ can be constructed by applying $\lo$.
Then $\forall$ can pick $a'\ll \bigsqcup A$ (in the case of powerset
lattices and a basis of singletons this amounts to choosing $a'\in A$)
and continue the game from there, asking $\exists$ to spell out why
$a'$ is also a logical formula.

If the lattice $\lt{B}$ is not continuous, the game can fail, in
particular $\exists$ may be able to win although $b\not\ll\mu\be$ (see
Example~\ref{ex:inst-co-cont} in the appendix). 
It could however be
the case that $\lt{B}$ is co-continuous and then we can use the game
presented next.

\begin{toappendix}
  \begin{example}\label{ex:inst-co-cont} 
    Let
    $\lt{B}=(\mathbb{N}_0\cup\{a, \omega, \omega+1\}, \sqsubseteq)$ be
    a lattice where $n\sqsubseteq \omega$ for all $n\in\mathbb{N}_0$,
    $0\sqsubseteq a\sqsubseteq\omega$, $\omega\sqsubseteq\omega+1$,
    and for any two elements in $\mathbb{N}_0$, $\sqsubseteq$
    coincides with $\le$.  Let $\jb{B}= \lt{B}\setminus\{0\}$.

    Then, $a\not\ll a$, since $a\sqsubseteq\bigsqcup \mathbb{N}_0$,
    however, there is no $n\in \mathbb{N}_0$ such that
    $a\sqsubseteq n$.  And therefore, $\lt{B}$ is not continuous,
    since $a$ does not arise as the join of elements that are
    way-below it.
  However, $\lt{B}$ is co-continuous. 
  We consider 
  the monotone and continuous function $\be:\lt{B}\rightarrow\lt{B}$,
  where $\be: 0\mapsto 0$, and $x\mapsto \omega+1$ otherwise.
  Since $a\not \ll\mu \be= 0$, $a$ meets the winning condition of
  player $\forall$ in the primal way-below game. However, if
  $\exists$ plays $a\ll a$, $\forall$ must answer with $b'$ such that
  $b'\ll a$.  The only element in $\lt{B}$ for which this holds is
  zero, however $0\not\in \jb{B}$. Therefore, player $\forall$ has
  no valid move and player $\exists$ wins although he should not.

  This example illustrates that the precondition of $\lt{B}$ being a
  continuous lattice is important in order to guarantee that the
  winning condition holds. However, if co-continuity holds, the game
  introduced in Section~\ref{sec:dual-game} will provide a winning
  condition. 
\end{example}
\end{toappendix}

\subsection{Dual Game}
\label{sec:dual-game}

We adjust the game in \cite{POPL19FPGames} for greatest fixpoints to
the dual setting by flipping the order and obtain another game
characterizing the least fixpoint. This gives us an alternative
perspective on games and witness generation.

\begin{definition}
  \label{def:dual-game}
  We consider the co-continuous lattice $\lt{B}$ with a meet basis
  $\mb{B}$ such that $\top\not\in \mb{B}$ and a monotone function
  $\be\colon\lt{B}\to \lt{B}$.  The dual game on
  $\be\colon\lt{B}\to\lt{B}$ follows the following rules.
  \begin{mdframed}
    \vspace{-0.2cm}
    \begin{center}
      \begin{tabular}{ c c l }
        \normalsize
        Position & \normalsize Player & \normalsize Moves \\ \hline
        \mystrut
        \normalsize $\du{b}\in \mb{B}$ & \normalsize $\exists$ &
        \normalsize $\du{d}\in\lt{B}$, such that
        $\be(\du{d}) \sqsubseteq \du{b}$ \\
        \normalsize $\du{d}\in \lt{B}$ & \normalsize $\forall$ &
        \normalsize $\du{b'}\in \mb{B}$ such that
        $\du{d} \wal \du{b'}$
      \end{tabular}
    \end{center}
    \vspace{-0.2cm}
  \end{mdframed}
  If a player cannot move, the opponent wins. Infinite games are won
  by $\exists$.
\end{definition}

Here $\exists$ takes the role of the defender, and $\forall$ the role
of the attacker. It holds that $\mu\be \sqsubseteq \du{b}$ iff
$\exists$ has a winning strategy when starting from $\du{b}$.

We show that $\forall$ has a finitary winning strategy that is in some
sense independent on the move of $\exists$. For this we require that
$\be$ is co-proper wrt.\ to the basis.

\begin{propositionrep}
  \label{prop:finite-winning-strategy-dual}
  Assume that $\lt{B}$ is co-continuous and that $\be$ is 
  co-proper wrt.\ $\mb{B}$.

  Fix $\du{b}\in \mb{B}$ with $\mu\be\not\sqsubseteq \du{b}$. Then
  there exists a finite set $F\subseteq \mb{B}$ such that -- for every
  move $\du{d}$ of $\exists$ -- $\forall$ can always choose some
  $\du{b'} \in F$ to win the game. In particular $\du{b'}\wag \du{d}$
  and $\cdegr(\du{b'}) < \cdegr(\du{b})$ for all $\du{b'}\in F$.  We
  denote $F$ by $\strat_{d,\be}^\forall(\du{b})$. Note that $\forall$
  can win in $\cdegr(\du{b})$ steps.
\end{propositionrep}

\begin{proof}
  Let $k=\cdegr(\du{b})$, which is a successor ordinal. We can infer
  that $\be^k(\bot)\not\sqsubseteq \du{b}$ and hence
  $\be^{k-1}(\bot)\not\in \be^{-1}(\mathop{\downarrow} \du{b})$.

  Now fix $\du{d}\in \be^{-1}(\mathop{\downarrow} \du{b})$. We know
  that $\be^{k-1}(\bot)\not\sqsubseteq \du{d}$ (since
  $\be^{-1}(\mathop{\downarrow} \du{b})$ is downward-closed and we
  could otherwise infer that
  $\be^{k-1}(\bot)\in \be^{-1}(\mathop{\downarrow} \du{b})$, a
  contradiction). Hence $\cdegr(\du{d})\le k-1$.

  Since $\lt{B}$ is co-continuous, 
  we can write $\du{d}$ as
  \[ \du{d} = \bigsqcap \{ \du{b'} \mid \du{b'}\in \mb{B},
    \du{b'}\wag \du{d} \}. \] 
    Since
  $\be^{k-1}(\bot)\not\sqsubseteq \du{d}$ there must be an element
  $\du{b'}\in \mb{B}$, $\du{b'}\wag \du{d}$ with
  $\be^{k-1}(\bot) \not\sqsubseteq \du{b'}$. Otherwise
  \[ \du{d} = \bigsqcap \{ \du{b'} \mid \du{b'}\in \mb{B},
    \du{b'}\wag \du{d} \} \sqsupseteq \be^{k-1}(\bot), \]
  a contradiction.
  This implies that $\cdegr(\du{b'}) \le k-1$.

  For each $\du{d}\in \be^{-1}(\mathop{\downarrow} \du{b})$ choose such a
  $\du{b'}$ and let $\mathcal{B}$ be the set of these elements. We
  observe that
  \[ \be^{-1}(\mathop{\downarrow} \du{b}) \subseteq
    \bigcup_{\du{b'}\in\mathcal{B}} \mathop{\wadownarrow} \du{b'}. \] Every
  set of the form $\mathop{\wadownarrow} \du{b'}$ is co-open. By assumption
  $\be^{-1}(\mathop{\downarrow} \du{b})$ is co-compact and hence every co-open
  cover of $\be^{-1}(\mathop{\downarrow} \du{b})$ contains a finite subcover.
  In other words there exists a finite subset $F\subseteq \mathcal{B}$
  such that
  \[ \be^{-1}(\mathop{\downarrow} \du{b}) \subseteq \bigcup_{\du{b'}\in F}
    \mathop{\wadownarrow} \du{b'} . \] So whenever $\exists$ plays $\du{d}$
  with $\be(\du{d})\sqsubseteq \du{b}$, there exists $\du{b'}\in F$
  such that $\du{b'}\wag \du{d}$ and
  $\cdegr(\du{b'}) = k-1 < k = \cdegr(\du{b})$. Then $\forall$ can
  choose this element $\du{b'}$ and wins after finitely many steps.
\end{proof}

\begin{example}[Strategy for the Dual Game]
  \label{ex:running-dual-game}
  We consider again our running example, i.e., the case of
  bisimilarity. We take the function $\be$ from
  Example~\ref{ex:running-intro} and the lattice
  $(\mathcal{P}(X\times X),\supseteq)$.
  The dual game on $\be$ corresponds to a coupling game
  \cite{b:coalgebraic-logic-games}. Given a basis element
  $\du{b} = \{(x_1,x_2)\}$, $\exists$ produces a ``coupling''
  $R\subseteq X\times X$ such that $(x_1,x_2)\in \be(R)$, which means
  every successor of $x_1$ must be paired with some successor of $x_2$
  and vice versa. Then $\forall$ picks one such pair in $R$, claims
  that it is not bisimilar and the game continues. Note whenever
  $x_1\not\sim x_2$ $\forall$ can always precompute a finitary winning
  strategy
  $\strat_{d,\be}^\forall(\du{b}) = \{\{(y_1,y)\} \mid y\in
  \mathit{succ}(x_2)\}$ 
  whenever we choose $y_1$ as in
  Example~\ref{ex:running-primal-way-below-game} as some state that
  has no bisimilar partner in $\mathit{succ}(x_2)$ (or vice versa).
\end{example}

\section{Transforming Witnesses and Winning Strategies} 
\label{sec:transformation}

\subsection{Auxiliary Functions}
\label{sec:auxiliary-functions}

Before we start to transform witnesses into winning strategies for the
attacker and vice versa, we first define some auxiliary
functions.

\begin{toappendix}
  \begin{lemma}
    \label{lem:auxiliary-functions-2}
    Let $\lt{L}$ be a continuous lattice and $\lt{B}$ a co-continuous
    lattice with meet basis $\mb{B}$.

    Let $\du{e}, \du{d}\in \lt{B}$ with
    $\du{e} \not\sqsubseteq \du{d}$. Then there exists
    $\du{b'}\in \mb{B}$ such that $\du{b'}\wag \du{d}$ and
    $\du{e}\not\sqsubseteq \du{b'}$.
  \end{lemma}

\begin{proof}
  We know that
  \[ \du{e}\not\sqsubseteq \du{d} = \bigsqcap \{\du{b'} \mid
    \du{b'}\in \mb{B}, \du{b'}\wag \du{d}\}, \] using the fact that
  $\lt{B}$ is co-continuous. Hence there must be $\du{b'}\in \mb{B}$
  with $\du{b'}\wag \du{d}$ and $\du{e} \not\sqsubseteq \du{b'}$.
\end{proof}
\end{toappendix}

We can show the existence of the following three functions specified
in the table below (see Lemmas~\ref{lem:auxiliary-functions-1} and
\ref{lem:auxiliary-functions-2}). For the first two lines we assume a monotone
function $f\colon\lt{L}\to \lt{L}$ and a finite set
$A\subseteq \jb{L}$. We will in particular instantiate to
$f=\mathsf{id}$ and $f=\lo$. We write $\aps$, $\ads$ instead of
$\auxp{\mathsf{id}}$, $\auxd{\mathsf{id}}$. For $f=\lo$ note that
$\alpha(\lo(\bigsqcup A))) = \be(\alpha(\bigsqcup A)) = \be(\bigsqcup
\alpha[A])$.

\medskip

\mbox{}\hspace{-0.6cm}
\begin{tabular}{|l|l|l|}
  \hline function & parameters & output \\ \hline $\auxp{f}$ &
  $b\in \jb{B}$, $A\subseteq_{\mathit{fin}} \jb{L}$, 
  $b\ll \alpha(f(\bigsqcup A))$ & $\auxp{f}(b,A) = a\in \jb{L}$ with
  $a\ll f(\bigsqcup A)$, $b\ll \alpha(a)$ \\
  $\auxd{f}$ & $\du{b}\in \mb{B}$,
  $A\subseteq_{\mathit{fin}} \jb{L}$,
  $\alpha(f(\bigsqcup A))\not\sqsubseteq \du{b}$ &
  $\auxd{f}(b,A) = \pr{a}\in \jb{L}$ with
  $\pr{a}\ll f(\bigsqcup A)$, $\alpha(\pr{a})\not\sqsubseteq \du{b}$ \\
  $\adwa$ & $\du{e},\du{d}\in \lt{B}$, $\du{e}\not\sqsubseteq \du{d}$
  & $\adwa(\du{e},\du{d}) = \du{b'}\in \mb{B}$ with
  $\du{b'}\wag \du{d}$ and $\du{e} \not\sqsubseteq \du{b'}$ \\
  \hline
\end{tabular}

\medskip

We assume that $\lt{L}$ is continuous lattice with join basis
$\jb{L}$, additionally in the first line the join basis $\jb{B}$
contains only irreducibles, in the second line $\mb{B}$ is a meet
basis and in the third line $\lt{B}$ must additionally be
co-continuous.

The first two functions
  have the task of picking a suitable witness from $\lt{L}$ from a set
  of joint witnesses. The last auxiliary functions chooses a basis
  element of $\lt{B}$ that explains why a given inequality does not
  hold.

\subsection{Primal Case: Transforming Strategies along the Galois
  Connection}

We now spell out how to use primal witnesses to obtain winning
strategies for $\exists$ in the primal way-below game (on
$\be\colon\lt{B}\to \lt{B}$), see
Section~\ref{sec:primal-way-below-game}. We furthermore show how such
winning strategies can be used to construct primal witnesses.
In this subsection we assume that $\lt{L}$, $\lt{B}$ are both
continuous lattices with join bases $\jb{L}$,
$\jb{B}$. Furthermore $\jb{L}$ must contain only way-below
irreducibles.

\begin{propositionrep}
  \label{Wp,lo,E to Wp,be,E}  
  Let $\lo$ be a continuous function, which implies that there is a
  finitary winning strategy $\strat_{p,\lo}^\exists$ for $\exists$ in
  the primal way-below game on $\lo\colon \lt{L}\to \lt{L}$
  (cf. Proposition~\ref{prop:winning-strategy-primal-way-below}).

  Given $b\in \jb{B}$ and a primal witness $\pr{a}$ for $b$
  ($b\ll\alpha(\pr{a})$), we choose the move of $\exists$ in 
  the primal way-below game on $\be:\lt{B}\to\lt{B}$ as 
  \[ d = \alpha(\bigsqcup \strat_{p,\lo}^\exists(\pr{a})). \] Let
  $b'\ll d$ be an answering move of $\forall$. Then
  $\pr{a}' = \aps(b',\strat_{p,\lo}^\exists(\pr{a}))$ is a primal
  witness for $b'$ and $\degr(\pr{a}') < \degr(\pr{a})$.
  We continue from $b'$ and $\pr{a}'$ and obtain
  a winning strategy for $\exists$ in the primal way-below game on
  $\be\colon\lt{B}\to \lt{B}$
  (cf. Definition~\ref{def:primal-way-below-game}) for all
  $b\in \jb{B}$ that have witnesses.
\end{propositionrep}

\begin{proof}
  Let $A = \strat_{p,\be}^\exists(\pr{a})$.  First observe that $d$ is
  a valid move of $\exists$ since $\strat_{p,\lo}^\exists$ is a valid
  strategy in the primal way-below game, hence
  $\lo(\bigsqcup \strat_{p,\lo}^\exists(\pr{a})) \gg \pr{a}$ and thus
  \[
    \be(d) = \be(\alpha(\bigsqcup
    \strat_{p,\lo}^\exists(\pr{a}))) = \alpha(\lo(\bigsqcup
    \strat_{p,\lo}^\exists(\pr{a}))) \sqsupseteq
    \alpha(\pr{a}) \gg b .
  \]
  Using the auxiliary function $\aps$ from
  Section~\ref{sec:auxiliary-functions} 
  we know that
  $\pr{a}' = \aps(b',A)$ satisfies $\pr{a}' \ll \bigsqcup A$ and
  $b'\ll \alpha(\pr{a}')$. That is, $\pr{a}'$ is a primal witness for
  $b'$.  Furthermore, using
  Lemma~\ref{lem:properties-deg}(\ref{lem:properties-deg-3}) and the
  fact that the winning strategy $\strat_{p,\lo}^\exists$
  satisfies the constraints of
  Proposition~\ref{prop:winning-strategy-primal-way-below} by
  assumption (i.e., the elements of $A$ have lower degree than
  $\pr{a}$), we can infer that
  \[ \degr(\pr{a}') \le \max_{\hat{a}\in A} \degr(\hat{a}) <
    \degr(\pr{a}). \]
  
  Hence, this results in a finite game, where eventually $\exists$
  plays $\bot$ and $\forall$ has no move left.
\end{proof}

Now we treat witness construction and assume that there exists a
finitary winning strategy $\strat_{p,\be}^\exists$ for $\exists$ in
the primal way-below game on $\be$
(cf. Definition~\ref{def:primal-way-below-game}) which assigns to each
$b\in \jb{B}$ a set $\strat_{p,\be}^\exists(b) \subseteq \jb{B}$. From
such a strategy we construct a witness $\pr{\mathit{wit}}(b)$ for a
given $b\ll \mu \be$.

\begin{toappendix}
\begin{proposition} 
  \label{Wp,B,E to Wp,L,E}
  Assume that $\be\colon\lt{B}\to\lt{B}$ is continuous, which implies
  the existence of a finitary winning strategy
  $\strat_{p,\be}^\exists$ for $\exists$ in the primal way-below game
  on $\be\colon\lt{B}\to\lt{B}$
  (cf. Proposition~\ref{prop:winning-strategy-primal-way-below}).

  Let $b\in \jb{B}$ with $b\ll \mu\be$ and we assume that
  $\strat_{p,\be}^\exists(b) = F\subseteq \jb{B}$ finite.  For
  each $b'\in F$ we fix a primal witness $\pr{\mathit{wit}}(b')$ and
  $\degr(\pr{\mathit{wit}}(b')) \le \degr(b')$.

  Let $\ell = \lo(\bigsqcup_{b'\in F} \pr{\mathit{wit}}(b'))$. Then
  there exists
  $\apc(b,\{\pr{\mathit{wit}}(b')\mid b'\in F\}) = \pr{a}\in
  \jb{L}$, $\pr{a}\ll \ell$, $\degr(\pr{a})\le \degr(b)$ such that
  $\pr{a}$ is a primal witness for $b$ and the corresponding finitary
  winning strategy is
  $\strat_{p,\lo}^\exists(\pr{a}) = \{\pr{\mathit{wit}}(b')\mid b'\in
  F\}$.
\end{proposition}

\begin{proof}
  We have
  \[ b \ll \be(\bigsqcup F) = \be(\bigsqcup_{b'\in F} b') \sqsubseteq
    \be(\bigsqcup_{b'\in F}\alpha(\pr{\mathit{wit}}(b'))) . \] The first
  inequality holds since $F$ is a winning strategy for the game and
  the second is true since $b'\ll \alpha(\pr{\mathit{wit}}(b'))$. We
  use the auxiliary function $\apc$ from
  Section~\ref{sec:auxiliary-functions} and obtain
  $\apc(b,\{\pr{\mathit{wit}}(b')\mid b'\in F\}) = \pr{a}$ with
  $\pr{a}\ll \ell$ and $b\ll\alpha(\pr{a})$, that is $\pr{a}$ is a
  primal witness for $b$.

  Furthermore
  \begin{align} \label{eqn:1}\pr{a}\ll \ell = \lo(\bigsqcup_{b'\in F}
    \pr{\mathit{wit}}(b')) \sqsubseteq \lo(\mu\lo) = \mu\lo, \end{align}
  hence $\pr{a}$ is a primal witness for $b$.

  We argue that $\degr(\pr{a})\le \degr(b)$: by assumption
  $\degr(\pr{\mathit{wit}}(b')) \le \degr(b') < \degr(b)$ for all
  $b'\in F$. The strict inequality follows from the choice of strategy
  according to
  Proposition~\ref{prop:winning-strategy-primal-way-below}.  Hence,
  using Lemma~\ref{lem:properties-deg}(\ref{lem:properties-deg-3}) we
  can infer that
  \begin{eqnarray}\label{eqn:2}
    \degr(\bigsqcup_{b'\in F} \pr{\mathit{wit}}(b')) = \max\{
    \degr(\pr{\mathit{wit}}(b')) \mid b'\in F\} \le \max\{
    \degr(b') \mid b'\in F\} < \degr(b) .
  \end{eqnarray}
  From (\ref{eqn:1}), we know that
  $\pr{a}\ll \lo(\bigsqcup_{b'\in F} \pr{\mathit{wit}}(b'))$, and from
  (\ref{eqn:2}) we have
  $\degr(\bigsqcup_{b'\in F} \pr{\mathit{wit}}(b'))<\degr(b)$.  Via
  Lemma~\ref{lem:properties-deg}(\ref{lem:properties-deg-2}) we obtain
  $\degr(\pr{a})\leq \degr(b)$.
  
  The construction will terminate eventually, since by construction of
  the winning strategy
  (cf. Proposition~\ref{prop:winning-strategy-primal-way-below})
  each $b'\in F$ has a smaller degree than $b$ ($\degr(b')<\degr(b)$).

  Furthermore $\strat_{p,\lo}^\exists(\pr{a})$ is a winning
  strategy for $\pr{a}$ since
  \[ \lo(\bigsqcup \strat_{p,\lo}^\exists(\pr{a})) =
    \lo(\bigsqcup \{\pr{\mathit{wit}}(b')\mid b'\in F\}) = \ell
    \gg \pr{a}
  \] and for each
  $\pr{a}' \ll \bigsqcup_{b'\in F} \pr{\mathit{wit}}(b')$ chosen by
  $\forall$ in the primal way-below game on $\lo$ we have that
  \[ \degr(\pr{a}') \le \degr(\bigsqcup_{b'\in F}
    \pr{\mathit{wit}}(b')) < \degr(b) \le \degr(\pr{a}), \] where the
  first inequality follows from
  Lemma~\ref{lem:properties-deg}(\ref{lem:properties-deg-1}) and last
  inequality from Lemma~\ref{lem:witness-degr}. Hence, the game
  will eventually terminate and $\exists$ will win.
\end{proof}
\end{toappendix}

From the construction one can also extract a finitary winning strategy
for the witness in the game on $\lo$ (see Proposition~\ref{Wp,B,E to
  Wp,L,E}). 

\begin{theoremrep}
  \label{thm:wit-of-StrtEpb}
  Assume that $\be\colon\lt{B}\to\lt{B}$ is continuous, which implies
  the existence of a finitary winning strategy
  $\strat_{p,\be}^\exists$ for $\exists$ in the primal way-below game
  on $\be\colon\lt{B}\to\lt{B}$
  (cf. Proposition~\ref{prop:winning-strategy-primal-way-below}).

  Given $b\in \jb{B}$ with $b\ll \mu\be$, we can compute the witness
  of $b$ inductively as follows:
  \[ \pr{\mathit{wit}}(b) =
    \apc(b,\pr{\mathit{wit}}[\strat_{p,\be}^\exists(b)]) . \] Then
  $\degr(\pr{\mathit{wit}}(b)) \le \degr(b)$ and this is a
  well-defined inductive definition with base case
  $\strat_{p,\be}^\exists(b) = \emptyset$.
\end{theoremrep}

\begin{proof}
  Corollary of Proposition~\ref{Wp,B,E to Wp,L,E}.
\end{proof}

\begin{example}[Witness Construction -- Primal]
  We study the construction of witnesses, also known as distinguishing
  formulas, in the running example. We are again working on
  $\mathbb{B} = \mathit{Rel}(X)$ (ordered by $\supseteq$) with a join
  basis consisting of co-singletons
  (Example~\ref{ex:running-basis-witness}, primal case).

  Assume that $x_1\not\sim x_2$ for $x_1,x_2\in X$ and let
  $b = \overline{\{(x_1,x_2)\}}$. We can choose some finitary strategy
  $\strat_{p,\be}^\exists$ for the primal way-below game on $\be$ (see
  e.g.\ Example~\ref{ex:running-primal-way-below-game}). Let
    $A = \mathit{wit}[\strat_{p,\be}^\exists(b)]$, i.e., we compute
    witnesses recursively.  We then obtain
    $\mathit{wit}(b) = \apc(b,A)$, where the auxiliary function $\apc$
    is defined as follows:
  \begin{itemize}
  \item If $A = \emptyset$ (which holds if $\degr(b) = 1$, hence
    $\strat_{p,\be}^\exists(b) = \emptyset$ and $\exists$ can play
    $\bigsqcup \emptyset = \bigcap \emptyset = X\times X$), then
    define $\apc(b,A) = \Diamond\mathit{true}\in \lo(\emptyset)$
    (where $\mathit{true}$ is the empty conjunction).
  \item Otherwise we know by construction that
    \begin{align*}
      \alpha(\lo(\bigsqcup A)) = \be(\alpha(A)) \gg b &\iff
      \be(\alpha(A)) \subseteq
      \overline{\{(x_1,x_2)\}} \\
      &\iff (x_1,x_2)\not\in \be(\alpha(A))
    \end{align*}
    Since $\alpha(A)$ suffices to distinguish $x_1,x_2$ (after
    applying $\be$), it must be the case that there exists
    $y_1\in\mathit{succ}(x_1)$ that is not related via
    $\alpha(A)$ to any state in $y_2\in \mathit{succ}(x_2)$
    (or vice versa). Hence for $y_1$ and each such $y_2$ there exists
    a predicate $P'\in A$ that separates them and we can define
    \[ \apc(b, A) = \Diamond \big(\bigcap_{P'\in A, y_1\in P'} P'
      \cap \bigcap_{P'\in A, y_1\not\in P'} \lnot P'\big), \] which
    contains $x_1$ (which has a successor $y_1$ satisfying all
    predicates) but not $x_2$ (where for each successor $y_2$ there
    exists a predicate not containing $y_2$).
  \end{itemize}
\end{example}

\subsection{Primal/Dual Case: Transforming Strategies along
  the Galois Connection}
\label{Sec:P/DCase:TransWSGC}

We now explain how to transform dual witnesses into winning strategies for
$\forall$ in the dual game on $\be\colon\lt{B}\to\lt{B}$.
In this subsection we assume that $\lt{L}$ is a continuous lattice
with join basis $\jb{L}$ and $\lt{B}$ is a co-continuous lattice
with meet basis $\mb{B}$.

\begin{propositionrep}\label{Wp,L,E to Wd,B,A}
  Assume that $\lo$ is continuous, which implies the existence of a
  finitary winning strategy $\strat_{p,\lo}^\exists$ for $\exists$ in
  the primal way-below game on $\lo\colon\lt{L}\to\lt{L}$
  (cf. Proposition~\ref{prop:winning-strategy-primal-way-below}).

  Let $\du{b}\in \mb{B}$ and let $a$ be a dual witness for $\du{b}$
  ($\alpha(a)\not\sqsubseteq \du{b}$). We let
  $A = \strat_{p,\lo}^\exists(a)$. Given a move $\du{d}\in\lt{B}$ by
  $\exists$, the $\forall$ player plays 
  \[ \du{b'} = \adwa(\alpha(\bigsqcup A),\du{d}), \] which has a dual
  witness $a' = \ads(\du{b'},A)$ with $\degr(a') < \degr(a)$. We
  continue from $\du{b'}$ and $a'$ and obtain a winning strategy for
  $\forall$ in the dual game on $\be\colon\lt{B}\to\lt{B}$ (see
  Definition~\ref{def:dual-game}).
\end{propositionrep}   

\begin{proof}
  Note that $\be(\du{d})\sqsubseteq \du{b}$. Since
  $\alpha(a)\not\sqsubseteq \du{b}$ we can infer that
  $\alpha(a)\not\sqsubseteq \be(\du{d})$. Furthermore
  $a \ll \lo(\bigsqcup A)$, otherwise
  $\strat_{p,\lo}^\exists$ would not be a winning strategy in the
  primal way-below game. Hence
  \[
    \alpha(a) \sqsubseteq  \alpha(\lo(\bigsqcup A)) =
    \be(\alpha(\bigsqcup A))  = \be(\bigsqcup \alpha[A]) 
    \not\sqsubseteq \be(\du{d})
  \]
  By monotonicity of $\be$ this implies:
  \[ \bigsqcup \alpha[A] \not\sqsubseteq \du{d} \] Hence we obtain
  $\du{b'} = \adwa(\bigsqcup \alpha[A],\du{d})$ which satisfies
  $\du{b'} \wag\du{d}$ and $\alpha[A]\not\sqsubseteq \du{b'}$ and is
  the chosen move of $\forall$.

  Furthermore $a' = \ads(\du{b'},A)$ satisfies $a'\ll \bigsqcup A$ and
  $\alpha(a') \not\sqsubseteq \du{b'}$, hence $a'$ is a dual witness
  for $\du{b'}$.

  From Lemma~\ref{lem:witness-degr} we obtain
  $\cdegr(\du{b'}) \le \degr(a')$ and we have
  \[ \degr(a') \le \degr(\bigsqcup A) = \max_{\hat{a}\in A}
    \degr(\hat{a}) < \degr(a), \] which follows from the choice of
  strategy. 

  When the game continues according to this strategy, $\forall$ will 
  play basis elements of smaller and smaller degree and finally win 
  because $\exists$ has no move left.   
\end{proof}

We will now explain how to construct a witness from a winning strategy
for $\forall$ in the dual game on $\be\colon\lt{B}\to\lt{B}$, using
the strategy from
Proposition~\ref{prop:finite-winning-strategy-dual}. From the
construction one can also extract a winning strategy for the witness
in the game on $\lo$ (cf. Proposition~\ref{Wd,be,A to
  Wp,lo,E}).

\begin{toappendix}
  \begin{proposition}\label{Wd,be,A to Wp,lo,E}
    Assume that $\be$ is co-proper wrt.\ $\mb{B}$, which implies the
    existence of a finitary winnning strategy $\strat_{d,\be}^\forall$
    for $\forall$ in the dual game on $\be\colon\lt{B}\to\lt{B}$.
    (cf. Proposition~\ref{prop:finite-winning-strategy-dual}).

  Let $\du{b}\in \mb{B}$ with $\mu\be\not\sqsubseteq\du{b}$ and let
  $F = \strat_{d,\be}^\forall(\du{b})$ be the set of possible answers
  of $\forall$.  For each $\du{b'}\in F$ we fix a dual witness
  $\mathit{wit}(\du{b'})$ with
  $\degr(\mathit{wit}(\du{b'}))\le \cdegr(\du{b'})$.

  Let $A = \{ \mathit{wit}(\du{b'})\mid \du{b'}\in F\}$. Then there
  exists $a = \adc(\du{b},A)\in \jb{L}$, $a\ll \lo(\bigsqcup A)$,
  $\degr(a)\le \cdegr(\du{b})$ such that $a$ is a dual witness for
  $\du{b}$ and the corresponding finitary winning strategy is
  $\strat_{p,\lo}^\exists(a) = \{\mathit{wit}(\du{b'})\mid
  \du{b'}\in F\}$.
\end{proposition}

\begin{proof}
  We define $\ell = \lo(\bigsqcup A) = \lo(\bigsqcup_{\du{b'}\in
    F} \mathit{wit}(\du{b'})))$.
  
  We set $k=\cdegr(\du{b})$ which implies that
  $\be^k(\bot)\not\sqsubseteq \du{b}$.

  Note also that
  $\alpha(\mathit{wit}(\du{b'})) \not\sqsubseteq \du{b'}$, since the
  $\mathit{wit}(\du{b'})$ are dual witnesses for the~$\du{b'}$.

  We first show that $\alpha(\ell)\not\sqsubseteq \du{b}$ and 
  observe that
  \begin{eqnarray*}
    && \alpha(\ell) = \alpha(\lo(\bigsqcup_{\du{b'}\in F}
    \mathit{wit}(\du{b'}))) = \be(\alpha(\bigsqcup_{\du{b'}\in F}
    \mathit{wit}(\du{b'}))) = \be(\bigsqcup_{\du{b'}\in F}
    \alpha(\mathit{wit}(\du{b'}))).
  \end{eqnarray*}

  We observe that $\alpha(\ell)\not\sqsubseteq \du{b}$ holds if
  and only if
  $\bigsqcup_{\du{b'}\in F} \alpha(\mathit{wit}(\du{b'}))
  \not\in\be^{-1}(\mathop{\downarrow} \du{b})$.
  Assume by contradiction that
  $\bigsqcup_{\du{b'}\in F} \alpha(\mathit{wit}(\du{b'}))
  \in\be^{-1}(\mathop{\downarrow} \du{b})$. In order to be valid
  strategy the sets $\mathop{\wadownarrow} \du{b'}$, $\du{b'}\in F$
  must form a finite open cover of
  $\be^{-1}(\mathop{\downarrow} \du{b})$
  (cf.~Proposition~\ref{prop:finite-winning-strategy-dual}),
  hence there exists $\du{b'}^*\in F$ such that
  $\bigsqcup_{\du{b'}\in F} \alpha(\mathit{wit}(\du{b'})) \wal
  \du{b'}^*$.  From this we get
  $\alpha(\mathit{wit}(\du{b'}^*)) \sqsubseteq \du{b'}^*$, a
  contradiction. Hence
  $\bigsqcup_{\du{b'}\in F} \alpha(\mathit{wit}(\du{b'}))
  \not\in\be^{-1}(\mathop{\downarrow} \du{b})$ and so
  $\alpha(\ell)\not\sqsubseteq \du{b}$.

  From this we obtain $a = \adc(\du{b},A)$ such that $a\ll \ell$
  and $\alpha(a) \not\sqsubseteq \du{b}$, hence $a$ is a dual witness
  for $\du{b}$.

  We argue that $\degr(a) \le \cdegr(\du{b})$: by assumption
  $\degr(\mathit{wit}(\du{b'})) \le \cdegr(\du{b'}) < \cdegr(\du{b})$.
  The second inequality follows from the choice of strategy according
  to Proposition~\ref{prop:finite-winning-strategy-dual}.  Using
  Lemma~\ref{lem:properties-deg}(\ref{lem:properties-deg-3}) we can
  infer that
  \begin{eqnarray*}
    \degr(\bigsqcup_{\du{b'}\in F} \mathit{wit}(\du{b'})) = \max\{
    \degr(\mathit{wit}(\du{b'})) \mid \du{b'}\in F\} \le \max\{
    \cdegr(\du{b'}) \mid \du{b'}\in F\} < \cdegr(\du{b}). 
  \end{eqnarray*}
  Hence, since $a\ll\ell=\lo(\bigsqcup_{\du{b'}\in F}
    \mathit{wit}(\du{b'}))$ and using
  Lemma~\ref{lem:properties-deg}(\ref{lem:properties-deg-2}), 
  we get
  \[ \degr(a) \le \degr(\bigsqcup_{\du{b'}\in F}
    \mathit{wit}(\du{b'})) +1 \le \cdegr(\du{b}). \]

  The construction will terminate eventually, since by construction of
  the winning strategy
  (cf. Proposition~\ref{prop:finite-winning-strategy-dual}) each
  $\du{b'}\in F$ has a smaller degree than $\du{b}$
  ($\cdegr(\du{b'})<\cdegr(\du{b})$).

  Furthermore $\strat_{p,\lo}^\exists(a)$ is a winning strategy for
  $a$ since whenever $\forall$ plays
  $\pr{a}' \ll \bigsqcup_{\du{b'}\in F} \mathit{wit}(\du{b'})$ in the
  primal way-below game on $\lo$ we have
  \[ \degr(\pr{a}') \le \degr(\bigsqcup_{\du{b'}\in F}
    \mathit{wit}(\du{b'})) < \cdegr(\du{b}) \le \degr(a), \] where the
  first inequality follows from
  Lemma~\ref{lem:properties-deg}(\ref{lem:properties-deg-1}) and the
  last equality from Lemma~\ref{lem:witness-degr}. Hence, the
  game will eventually terminate and $\exists$ will win.

  In the base case ($F=\emptyset$) note that $\ell = \lo(\bot)$. We 
  choose $a\ll\lo(\bot)$ as witness and the construction terminates 
  since $a\ll \mu\lo$ can be witnessed by playing $\bigsqcup \emptyset
  = \bot$.
\end{proof}
\end{toappendix}

\begin{theoremrep}
  \label{thm:stratFdbToWit}
  Assume that $\be$ is co-proper wrt.\ $\mb{B}$, which implies the
  existence of a finitary winnning strategy $\strat_{d,\be}^\forall$
  for $\forall$ in the dual game on $\be\colon\lt{B}\to\lt{B}$.
  (cf. Proposition~\ref{prop:finite-winning-strategy-dual}).

  Given $\du{b}\in \mb{B}$ with $\mu\be\not\sqsubseteq\du{b}$,
  we can compute the witness of $\du{b}$ inductively as:
  \[ \mathit{wit}(\du{b}) =
    \adc(\du{b},\mathit{wit}[\strat_{d,\be}^\forall(\du{b})]). \] Then
  $\degr(\mathit{wit}(\du{b}))\le \degr(\du{b})$ and this is a
  well-defined inductive definition with base case
  $\strat_{d,\be}^\forall(\du{b}) = \emptyset$.
\end{theoremrep}

\begin{proof}
  Corollary of Proposition~\ref{Wd,be,A to Wp,lo,E}.
\end{proof}

\begin{example}[Witness Construction -- Dual]
  In the running example, we choose a meet basis (for the order
  $\supseteq$) for $\mathit{Rel}(X)$ with singletons as in
  Example~\ref{ex:running-basis-witness}. 

  Now fix $\du{b} = \{(x_1,x_2)\}$ with $x_1\not\sim x_2$ and we
    choose some finitary strategy $\strat_{d,\be}^\forall$ for
    the dual game on $\be$ (see e.g.\
    Example~\ref{ex:running-dual-game}). By construction
    $\mathit{wit}(\du{b}) = \adc(\du{b},A)$, where
    $A = \mathit{wit}[\strat_{d,\be}^\forall(\du{b})]$ and $\adc$ is
    defined as follows:
    \begin{itemize}
    \item If $A=\emptyset$ (which holds if $\cdegr(\du{b}) = 1$), we again
      set $\adc(\du{b},A) = \Diamond \mathit{true}$.
    \item Otherwise
      \[
        \alpha(\lo(\bigsqcup A)) = \be(\alpha(A)) \not\sqsubseteq
        \du{b} \iff \be(\alpha(A)) \not\supseteq
        \{(x_1,x_2)\} \iff (x_1,x_2)\not\in \be(\alpha(A))
      \]
      As in the primal case we determine a state
      $y_1\in \mathit{succ}(x_1)$ that is not related (via
      $\alpha(A)$) to any state in $\mathit{succ}(x_2)$ (or vice
      versa) and define:
      \[ \adc(\du{b},A) = \Diamond \big(\bigcap_{P'\in A, y_1\in
          P'} P' \cap \bigcap_{P'\in A, y_1\not\in P'} \lnot
        P'\big). \]
    \end{itemize}
\end{example}

\section{Case Studies}
\label{sec:instantiations}

\subsection{Behavioural Metrics}
\label{sec:inst-beh-metrics}

We now consider the case of behavioural metrics, where we measure the
distance of two states. Here we focus on probabilistic transition
systems. The construction of distinguishing formulas in this setting
was presented previously in \cite{RadyBreugelExplProbBisim}, here we
show how a similar construction arises as a
special case of our theory. Note that the instantiation of
probabilistic transition systems to the Galois connection approach is
new, it was not studied in \cite{HMLGalois}.

We consider labelled Markov chains $(X, \delta, \ell)$ consisting of a
\emph{finite} state space $X$, a probabilistic transition function
$\delta\colon X\to \mathcal{D}(X)$ (where $\mathcal{D}(X)$ is the set
of probability distributions over $X$) and a labelling function
$\ell\colon X\to \Lambda$.

\subparagraph*{Lattices, Functions and Galois Connection:}~

\noindent\emph{Behaviour:} On the behaviour side we fix the lattice
$\lt{B} = (\mathit{Dist}(X),\le)$ (where $\mathit{Dist}(X)$ denotes
the set of distance functions over $X$, i.e., functions of the form
$d\colon X\times X\to [0,1]$). The way-below and way-above relations
$\ll$, $\wal$ are the pointwise $<$-orders (in addition $0\ll 0$,
$1\wal 1$), see Example~\ref{ex:way-below-above}. The behaviour function
$\be$ is characterized as: given $d\in \mathit{Dist}(X)$ it holds that
\begin{eqnarray*}
  \be(d)(x_1,x_2) & = &
  \begin{cases}
    \mathcal{K}(d)(\delta(x_1),\delta(x_2)) & \text{if
      $\ell(x_1)=\ell(x_2)$} \\
    1 & \text{otherwise}
  \end{cases}
\end{eqnarray*}
where $\mathcal{K}$ is the (price-function based) Kantorovich lifting
\cite{v:optimal-transport,bbkk:coalgebraic-behavioral-metrics} that
transform a distance on $X$ into a distance on $\mathcal{D}(X)$. In
terms of the Galois connection this is defined as
$\mathcal{K} = \alpha\circ\bigcirc\circ \gamma$. Here
$\bigcirc f(x) = \mathbb{E}_{\delta(x)}[f]$ where
$\mathbb{E}_{p}[f] = \sum_{x\in X} f(x)\cdot p(x)$ determines the
expectation of random variable $f$ under probability distribution $p$.
The function $\be$ is obviously monotone (see also
\cite{RadyBreugelExplProbBisim}).

Hence $\mu\be(x_1,x_2)$ denotes the behavioural distance of $x_1,x_2$,
based on the Kantorovich lifting.

\medskip

\noindent\emph{Logic:} The lattice on the logic side is
$\lt{L} = (\mathcal{P} ([0,1]^X),\subseteq)$ (sets of random
variables). The way-below relation for powerset lattices is defined in
Example~\ref{ex:way-below-above}.  The logic function is based on the
operators of \cite{RadyBreugelExplProbBisim}. We define
$\lo\colon \mathcal{P}([0,1]^X)\to \mathcal{P}([0,1]^X)$ with
\[
  \lo(\mathcal{F}) = \{ [a] \mid a\in\Lambda\} \cup \{ \bigcirc f \mid
  f\in \mathit{cl}(\mathcal{F}) \}
\]
where $[a]\colon X\to [0,1]$ is defined as $[a](x) = 1$ if $\ell(x)=a$
and $0$ otherwise. Furthermore $\mathit{cl}(\mathcal{F})$ is the
closure of the set $\mathcal{F}\subseteq [0,1]^X$ under the operators
$1-f$, $f\ominus q$ ($q\in \mathbb{Q}\cap [0,1]$, $\ominus$ is the
modified subtraction) and $\max(f,g)$ for functions
$f,g\colon X\to [0,1]$.

\medskip

\noindent\emph{Galois Connection:} The Galois connection is given as
follows:
\[
  \alpha(\mathcal{F}) = (x_1,x_2)\mapsto \sup_{f\in \mathcal{F}}
   |f(x_1)-f(x_2)| \qquad \gamma(d) = \{ f\colon X\to
  [0,1] \mid \text{$f$ non-expansive wrt. $d$} \}
\]
Given $d\in \mathit{Dist}(X)$, a \emph{non-expansive} function
$f\colon X\to [0,1]$ wrt.\ $d$ must satisfy
$|f(x_1)-f(x_2)| \le d(x_1,x_2)$.

\begin{toappendix}
  We argue that all the requirements are satisfied. First, both
  lattices are continuous and co-continuous.

  Second, we have to check whether $\alpha\circ \lo = \be\circ
  \mathop{\alpha}$.

  \begin{proposition}
    In the setting of behavioural metrics, it holds that
    $\be\circ \mathop{\alpha} = \alpha \circ \lo$.
  \end{proposition}
  
  \begin{proof}
    We have to show that
    \[ \be\circ \mathop{\alpha}(\mathcal{F})(x_1,x_2) = \alpha \circ
      \lo(\mathcal{F})(x_1,x_2) \] for $\mathcal{F}\subseteq [0,1]^X$
    and $x_1,x_2\in X$.

    Whenever $\ell(x_1)\neq \ell(x_2)$ we obtain $1$ on the left-hand
    side by definition of $\be$. Furthermore the logic function $\lo$
    provides the formula $[a]$ where $a=\ell(x_1)$ that evaluates on
    $x_1$ to $1$ and on $x_2$ to $0$, which witnesses distance $1$ on
    the right-hand side.

    Now let $\ell(x_1)=\ell(x_2)$. Then we have -- via the definition
    of the Kantorovich lifting $\mathcal{K}$ -- on the left-hand side:
    \begin{eqnarray*}
      \be\circ\alpha(\mathcal{F})(x_1,x_2) & = & \mathcal{K}\circ
      \alpha(\mathcal{F})(x_1,x_2) = \alpha\circ\bigcirc\circ
      \gamma\circ\alpha(\mathcal{F})(x_1,x_2)
    \end{eqnarray*}
    Note that $\gamma\circ\alpha(\mathcal{F})$ is the set of all
    non-expansive functions of the metric that is generated by
    $\mathcal{F}$.  On the right-hand side we can apply the operators
    of the logic to $\mathcal{F}$ to approximate every such
    non-expansive function with arbitrary precision
    \cite{bw:behavioural-pseudometric,km:bisim-games-logics-metric}. Since
    -- due to the non-expansiveness of the operators -- we never
    over-estimate distances and $\bigcirc$ is continuous, we can infer
    that $\alpha \circ \lo(\mathcal{F})(x_1,x_2)$ results in the same
    value than the one on the left-hand side. 
  \end{proof}
\end{toappendix}

\subparagraph*{Primal Case:}
We now construct witnesses in the primal case. We first observe that
$\be$ is indeed continuous.

\medskip

\noindent\emph{Basis:} As join basis of $\lt{L}$ (consisting of way-below
irreducibles) we consider singleton sets $\{f\}$ where
$f\colon X\to [0,1]$, while the join basis of $\lt{B}$ contains all
distance functions $d_{x_1,x_2}^c$ (for $x_1\neq x_2$, $c>0$) that are
defined as
\[ d_{x_1,x_2}^c(y_1,y_2) =
  \begin{cases}
    c & \text{if $(x_1,x_2) = (y_1,y_2)$} \\
    0 & \text{otherwise}
  \end{cases}
\]

\noindent\emph{Strategy computation (cf. Prop.~\ref{prop:winning-strategy-primal-way-below}):} Let $b = d_{x_1,x_2}^c$ be a
basis element that is way-below $\mu\be$ (i.e.,
$c < \mu\be(x_1,x_2)$), for which we determine the strategy of
$\exists$. Let $k=\degr(b)$ and define $d_k=\be^k(0)$ as the
$k$-th iterate in the Kleene iteration.

Whenever $k=1$ we have $\ell(x_1)\neq \ell(x_2)$ and $\exists$ can
play $\emptyset$.

Otherwise we rely on the coupling characterization of the Kantorovich
lifting \cite{v:optimal-transport}. Given two probability
distributions $p,q\in\mathcal{D}(X)$, a \emph{coupling} of
$p,q$ is a probability distribution $C\in \mathcal{D}(X\times X)$
with $p,q$ as marginals, i.e., for all $x_1\in X$:
$\sum_{x_2\in X} C(x_1,x_2) = p(x_1)$ and for all $x_2\in X$:
$\sum_{x_1\in X} C(x_1,x_2) = q(x_2)$. We denote the set of
couplings of $p,q$ as $\Gamma(p,q)$. Then -- for a
pseudometric $d$ -- we can spell out the Kantorovich lifting as
\begin{eqnarray*}
  \mathcal{K}(d)(x_1,x_2) & = & \inf \{\sum_{y_1,y_2\in X}
    C(y_1,y_2)\cdot d(y_1,y_2) \mid C\in
    \Gamma(\delta(x_1),\delta(x_2)) \} \\
    & = & \inf \{ \mathbb{E}_C[d] \mid
    C\in \Gamma(\delta(x_1),\delta(x_2)) \} .
\end{eqnarray*}
The set of all couplings form a polytope and -- as is well-known in
linear programming -- the infimum above is a minimum and is achieved
in one of the (finitely many) vertices of the coupling polytope. We
denote the set of vertices by $\Gamma_V(\delta(x_1),\delta(x_2))$ and
can replace $\Gamma(\delta(x_1),\delta(x_2))$ in the equation above by
$\Gamma_V(\delta(x_1),\delta(x_2))$ and $\inf$ by $\min$.

Let $\mathit{succ}(x_i)$ be the states reachable from $x_i$ with
non-zero probability. For each pair
$y_1\in \mathit{succ}(x_1),y_2\in \mathit{succ}(x_2)$ determine
constants $c_{y_1,y_2}\in[0,1]$ -- which may equal $0$ -- that satisfy
the following inequalities for some coupling $C$ of the successor
sets:
\[ \sum_{y_1,y_2} C(y_1,y_2)\cdot c_{y_1,y_2} > c \qquad\qquad
  d_{k-1}(y_1,y_2) > c_{y_1,y_2} \mbox{ or } c_{y_1,y_2}=0. \] We also have to include
inequalities characterizing pseudometrics (reflexivity, symmetry,
triangle inequality) so that we obtain a pseudometric, ensuring that
the coupling-based definition coincides with the Kantorovich lifting
and we obtain a valid strategy.

Then collect all basis elements $d_{y_1,y_2}^{c_{y_1,y_2}}$ (for
$c_{y_1,y_2}\neq 0$) as the finitary strategy
$\strat_{p,\be}^\exists(d_{x_1,x_2}^c)$.

This is reminiscent of the game introduced in
\cite{vjwb:explainability-probabilistic-game} where couplings are used
as policies explaining the distance.

\medskip

\noindent\emph{Auxiliary functions
  (cf. Sct.~\ref{sec:auxiliary-functions}):} Let $b = d_{x_1,x_2}^c$
and $A \subseteq [0,1]^X$. We define:
\begin{itemize}
\item $\aps(b,A) = f$, where $f\in A$ such that $|f(x_1)-f(x_2)| > c$.
\item $\apc(b,A)$: if $\ell(x_1)\neq \ell(x_2)$ choose $[a]\in\lo(A)$.
  Otherwise, since
  \[ c < \be(\bigsqcup \alpha[A])(x_1,x_2) =
    \mathcal{K}(\bigsqcup\alpha[A])(x_1,x_2) =
    \alpha\circ\bigcirc\circ\gamma(\bigsqcup\alpha[A])(x_1,x_2) \]
  there exists a price function $f$ that is non-expansive wrt.\
  $\bigsqcup\alpha[A]$ and $|\bigcirc f(x_1) - \bigcirc f(x_2)| >
  c$. Now approximate $f$ from $A$ by using the logical operators and
  apply $\bigcirc$ (for more details see
  \cite{bw:behavioural-pseudometric,km:bisim-games-logics-metric,RadyBreugelExplProbBisim})
  to obtain a logical formula in $\lo(\bigsqcup A)$ that witnesses
  that the distance of $x_1,x_2$ is at least $c$.
\end{itemize}

Combined this gives us a construction similar to
\cite{RadyBreugelExplProbBisim}. There the price function is chosen
directly while we first determine a strategy. In
\cite{RadyBreugelExplProbBisim} it has been observed that one can not
always construct a formula that witnesses the exact distance between
$x_1,x_2$. This problem is avoided here, since we are interested in
certifying strict lower bounds.

\subparagraph*{Dual Case:}~

\noindent\emph{Basis:} In the dual case we choose a meet basis for
$\mathit{Dist}(X)$ that contains all distance functions
$\du{d}_{x_1,x_2}^c$ (for $x_1\neq x_2$, $c<1$) that are defined as
\[ \du{d}_{x_1,x_2}^c(y_1,y_2) =
  \begin{cases}
    c & \text{if $(x_1,x_2) = (y_1,y_2)$} \\
    1 & \text{otherwise}
  \end{cases}
\]
It can be shown that $\be$ is co-proper wrt.\ the basis (for more
details see the appendix).

\begin{toappendix}
  \begin{proposition}
    In the case of behavioural metrics the map $\be$ is co-proper
    wrt.\ the chosen meet basis (dual case).
  \end{proposition}

\begin{proof}
  Let $\du{b} = \du{d}_{x_1,x_2}^c$ an element of the meet basis. We
  now argue that $\be^{-1}(\mathop{\downarrow} \du{b})$ for a basis element
  $\du{b}$ is always a compact set. It is a closed set in the
  Euclidean topology on $[0,1]^{X\times X}$ (that is generated by
  either the Euclidean or the sup metric). And since $\be$ is
  non-expansive wrt.\ the $\sup$ metric, hence continuous, we can
  infer that $\be^{-1}(\mathop{\downarrow} \du{b})$ is closed wrt.\ the
  Euclidean topology on $[0,1]^{X\times X}$. The space
  $[0,1]^{X\times X}$ is compact, from which we can infer that
  $\be^{-1}(\mathop{\downarrow} \du{b})$ is compact in the Euclidean
  topology. Since every Scott-open set is also open in the Euclidean
  topology it is also compact in the Scott topology.
\end{proof}
\end{toappendix}

\medskip

\noindent\emph{Strategy computation
  (cf. Prop.~\ref{prop:winning-strategy-primal-way-below}):}
The computation of the strategy
$\strat_{d,\be}^\forall(\du{d}_{x_1,x_2}^c)$ in the dual case works
analogously to the primal case. This is due to the fact that
$d_{x_1,x_2}^c \ll \mu\be$ iff
$\mu\be\not\sqsubseteq \du{d}_{x_1,x_2}^c$ and the orders $\ll$,
$\wal$ coincide (on elements different from $0,1$). Furthermore the
requirement that for all moves $d$ of $\exists$ with
$\be(d)\sqsubseteq b$ there exists a basis element $\du{b'}$ of the
chosen finitary strategy such that $d\wal \du{b'}$ can also be ensured
via the inequality involving couplings as above. 

\medskip

\noindent\emph{Auxiliary functions
  (cf. Sct.~\ref{sec:auxiliary-functions}):} The auxiliary functions
$\ads$, $\adc$ can also be defined analogously to the primal case. In
addition:
\begin{itemize}
\item $Z(\du{e},\du{d})$: $\du{e},\du{d}\in\mathit{Dist}(X)$ with
  $\du{e}\not\sqsubseteq\du{d}$, i.e., there exists $x_1,x_2\in X$
  with $\du{e}(x_1,x_2) > \du{d}(x_1,x_2)$. In this case choose $c$
  such that $\du{e}(x_1,x_2) > c > \du{d}(x_1,x_2)$ and return
  $\du{d}_{x_1,x_2}^c$.
\end{itemize}

\subsection{Termination Probabilities in Markov Chains}
\label{sec:inst-term-mc}

In this section, we consider unlabelled Markov chains
\cite{gs:markov-chains,bk:principles-mc} and witness termination
probabilities of states. We fix a Markov chain $(X, T, \delta)$ which
has a finite state space $X$, a subset of terminal states
$T\subseteq X$ (which do not have outgoing transitions) and a
probabilistic transition function
$\delta\colon X\backslash T\to \mathcal{D}(X)$, where $\mathcal{D}(X)$
is a set of discrete probability distributions. The termination
probability of a state $x\in X$ is the probability that a run starting
from $x$ will eventually terminate in a state in $T$.

\subparagraph*{Lattices, Functions and Galois Connections:}~

\noindent\emph{Behaviour:} On the behaviour side we use the lattice $\lt{B}=[0,1]^X$,
with function $\be\colon\lt{B}\to \lt{B}$ defined as:
\[ \be(f)(x) = \begin{cases}
    1 & \text{if $x\in T$} \\
    \sum_{y\in X} \delta(x)(y) \cdot f(y) & \text{otherwise}.
  \end{cases} \] This function is clearly monotone and its least
fixpoint $\mu \be$ assigns to each state its termination probability.

\medskip

\noindent\emph{Logic:} Witnesses are \emph{trees}, where a tree is either a
terminal node $t\in T$, or, of the form
$x\rightarrow \mathit{Tr}_1,\dots, \mathit{Tr}_k$ for
$x\in X\setminus T$ and where the $\mathit{Tr}_i$ are trees.  Every
tree has a root, defined by $\mathit{root}(t)=t$, for $t\in T$, and
$\mathit{root}(x\rightarrow \mathit{Tr}_1,\dots, \mathit{Tr}_k)= x$.
We require that in a tree
$x\rightarrow \mathit{Tr}_1,\dots, \mathit{Tr}_k$ the children
$\mathit{Tr}_1,\dots, \mathit{Tr}_k$ all have different roots.  The
set of all trees will be denoted as $\mathit{Trees}$.  For a set of
trees $\mathcal{T}\subseteq \mathit{Trees}$, we write
$\mathcal{T}_x = \{\mathit{Tr}\in \mathcal{T} \mid
\mathit{root}(\mathit{Tr})= x\}$.  The degree of a tree corresponds to
its height, where the height of $t\in T$ is~$1$.

We consider a map $\mathit{pt}: \mathit{Trees} \rightarrow [0,1]$,
where $\mathit{pt}(\mathit{Tr})$ under-estimates the termination
probability from $\mathit{root}(\mathit{Tr})$, based on the paths in
$\mathit{Tr}$.
\[\mathit{pt}(\mathit{Tr}) =
\begin{cases}
  1 & \mbox{if $\mathit{Tr} = t\in T$} \\
  \sum_{i=1}^{k} \delta(x)(\mathit{root}(\mathit{Tr}_i)) \cdot
  pt(\mathit{Tr}_i) & \mbox{if $\mathit{Tr} = x\rightarrow \mathit{Tr}_1, \dots,
  \mathit{Tr}_k$}
\end{cases}\]

On the logic side we use the lattice
$\lt{L}= \mathcal{P}(\mathit{Trees})$ (with inclusion order) with
functions $\lo\colon \lt{L}\rightarrow \lt{L}$:
\[ \lo(\mathcal{T}) = \{x\rightarrow \mathit{Tr}_1,\dots,\mathit{Tr}_k\mid x\in
  X\backslash T, \mathit{Tr}_i\in \mathcal{T},
  \mathit{root}(\mathit{Tr}_i)\not=\mathit{root}(\mathit{Tr}_j) \text{
    for } i\not=j \} \cup \{t\in T\},\] the least fixpoint of which is
the set of all trees. Since the trees are finitely branching, it can
easily be seen that $\lo$ is continuous, the same holds for $\be$.

\medskip

\noindent\emph{Galois Connection:} The Galois connection is given as
follows:
\[\alpha(\mathcal{T})= \lambda x. \sup_{\mathit{Tr}\in \mathcal{T}_x}
  \mathit{pt}(\mathit{Tr}) \qquad \gamma(f)= \{ \mathit{Tr}\in
  \emph{Trees} \mid \mathit{pt}(\mathit{Tr})\leq
  f(\mathit{root}(\mathit{Tr}))\}. \] Here, $\alpha$ maps a set of
trees $\mathcal{T}$ to a function that provides a lower bound
for the termination probability of a
node $x$, based on the trees in $\mathcal{T}_x$.  On the other hand,
$\gamma$ maps a function to all trees $\mathit{Tr}$ which induce a
lower value for the respective assignment.

\begin{toappendix}
  We will show that $\mathit{pt}$ in fact under-approximates the
  termination probability.

\begin{lemma}
  For each tree $\mathit{Tr}\in\mathit{Trees}$ it holds that
  $\mathit{pt}(\mathit{Tr}) \le \mu \be(\mathit{root}(\mathit{Tr}))$.
\end{lemma}

\begin{proof}
  We show by induction on $i$ that whenever
  $\degr(\mathit{Tr})\le i$ for a tree $\mathit{Tr}$, then
  \[ \mathit{pt}(\mathit{Tr}) \le
    \mu\be(\mathit{root}(\mathit{Tr})). \]
  \begin{itemize}
  \item $i=1$: In this case $\mathit{Tr} = t\in T$ and we obtain
    $\mathit{pt}(\mathit{Tr}) = 1$. On the other side of the
    inequality we also get $\mu\be(t) = \be(\mu\be)(t) = 1$.
  \item $i\to i+1$: In this case either $\mathit{Tr} = t\in T$ and the
    inequality holds for the same reasons as in the previous case. Or
    $\mathit{Tr} = x\to \mathit{Tr}_1,\dots,\mathit{Tr}_k$,
    $x\not\in T$, where $\degr(\mathit{Tr}_j) \le i$ for each
    $j\in\{1,\dots,k\}$. We obtain:
    \begin{eqnarray*}
      \mathit{pt}(\mathit{Tr}) & = & \sum_{i=1}^k
      \delta(x)(\mathit{root}(\mathit{Tr}_i))\cdot
      \mathit{pt}(\mathit{Tr}_i) \\
      & \le & \sum_{i=1}^k
      \delta(x)(\mathit{root}(\mathit{Tr}_i))\cdot
      \mu\be(\mathit{root}(\mathit{Tr}_i)) \\
      & \le & \sum_{y\in X} \delta(x)(y)\cdot
      \mu\be(y) \\
      & = & \be(\mu \be)(x) = \mu\be(x)
    \end{eqnarray*}
    where the first inequality uses the induction hypothesis and the
    second inequality holds since the trees $\mathit{Tr}_i$ all have
    different roots. \qedhere
  \end{itemize}
\end{proof}

Second, we show that the requirement for fixpoint preservation
($\alpha(\mu\lo) = \mu\be$) indeed holds.

\begin{proposition}
  In the setting of termination probabilities in Markov chains, it
  holds that $\be\circ \mathop{\alpha} = \alpha \circ \lo$.
\end{proposition}

\begin{proof}
  Let $\mathcal{T}$ be a set of trees and let $f = \alpha(\mathcal{T})$, where
  $f(x) = \sup_{\mathit{Tr}\in \mathcal{T}_x} \mathit{pt}(\mathit{Tr})$.

  We define $X' = \{x\in X \mid \mathcal{T}_x\neq\emptyset\}$ as those states
  that have witness trees in $\mathcal{T}$. Note that $f(x) = 0$ whenever
  $x\not\in X'$. For $t\in T$ we have $\be(f)(t) = 1$. And for
  $x\in X\backslash T$ we obtain:
  \begin{align*}
    \be(f)(x) & = \sum_{y\in X} \delta(x)(y)\cdot f(y)  \\
    & = \sum_{y\in X'} \delta(x)(y)\cdot f(y)  \\
    & = \sum_{y\in X'}  \delta(x)(y)\cdot \sup_{\mathit{Tr}\in \mathcal{T}_x} \mathit{pt}(\mathit{Tr})  \\
    & = \sup_{\substack{g\colon X'\rightarrow \mathcal{T} \\ g(y)\in \mathcal{T}_y}}
    \sum_{y\in X'} \delta(x)(y)\cdot \mathit{pt}(g(y))
  \end{align*}
  The last equality uses a supremum over all choice functions $g$,
  which for any state $y\in X'$ choose a tree in $\mathcal{T}$ with
  $y$ as the root.

  For the other side we set $f' = \alpha\circ\lo(\mathcal{T})$ where $\lo (\mathcal{T})$
  is defined as above. For a state $t\in T$ we have $f'(t) = 1$ since
  all trees with root $t$ evaluate to $1$. Now let
  $x\in X\backslash T$:
  \begin{align*}
    f'(x) & = \sup_{\substack{\mathit{Tr}\in \lo (\mathcal{T}) \\
        \mathit{root}(\mathit{Tr})=x}}
    \mathit{pt}(\mathit{Tr}) \\
    & = \sup_{\substack{\mathit{Tr}_1,\dots, \mathit{Tr}_k\in \mathcal{T} \\
        \mathit{root}(\mathit{Tr}_j)\neq
        \mathit{root}(\mathit{Tr}_\ell)\text{ if $j\neq\ell$}}}
    \sum_{i=1}^{k} \delta(x)(\mathit{root}(\mathit{Tr}_i))\cdot \mathit{pt}(\mathit{Tr}_i) \\
    & = \sup_{\substack{g\colon X'\rightarrow \mathcal{T} \\ g(y)\in \mathcal{T}_y}}
    \sum_{y\in X'} \delta(x)(y)\cdot \mathit{pt}(g(y))
  \end{align*}
  In the last line we again uses choice functions as above. The last
  equality holds since the supremum is reached in a sum where every
  state $y\in X'$ occurs as a root of a tree $g(y)$.

  Hence we have $\be(f)(x) = f'(x)$ for all $x\in X$ and therefore
  $\be\circ \mathop{\alpha} = \alpha \circ \lo$.
\end{proof}
\end{toappendix}

We now construct witnesses that certify termination probabilities for
the primal and for the dual case.

\subparagraph*{Primal case:}~

\noindent\emph{Basis:} for the join basis of $\lt{L}$ 
we take the singletons $\{\mathit{Tr}\}$, where
$\mathit{Tr}\in \mathit{Trees}$ is a tree. For the join basis of
$\lt{B}$ we take all functions $f^{c}_x$ (for $x\in X$, $c>0$) where
\[ f^{c}_x (y) =
  \begin{cases}
    c & \text{ if } x=y \\
    0 & \text{otherwise}
  \end{cases} \]

\noindent\emph{Strategy computation
  (cf. Prop.~\ref{prop:winning-strategy-primal-way-below}):}
let $k = \degr(f_x^c)$. If $k=1$, $\exists$ plays the empty
  set. Otherwise we have to solve the following inequalities (where
$d_k = \be^k(0)$) in order to obtain values $c_i\in[0,1]$ for the
successors $x_1,\dots,x_n$ of $x$:
\[
  \sum_{i=1}^{n} \delta(x)(x_i)\cdot c_i > c \qquad\qquad d_{k-1}(x_i)
  > c_i \mbox{ or } c_i=0
\]
Then collect all such $f_{x_i}^{c_i}$ where $c_i\neq 0$ as finitary
strategy $\strat_{p,\be}^\exists(f_x^c)$.

\medskip

\noindent\emph{Auxiliary functions
  (cf. Sct.~\ref{sec:auxiliary-functions}):} Let $b = f_x^c$
and $A \subseteq \mathit{Trees}$. We define:
\begin{itemize}
\item $\aps(b,A) = \mathit{Tr}$ where $\mathit{Tr}\in A$ is a tree
  such that $\mathit{pt}(\mathit{Tr}) > c$.
\item $\apc(b,A)$: if $A=\emptyset$ it must hold that $x\in T$, hence
  choose $x$. Otherwise choose
  $\{\mathit{Tr}_1,\dots,\mathit{Tr}_n\}\subseteq A$ such that each
  tree $\mathit{Tr}_i$ witnesses a maximal termination probability for
  one of the successors of $x$ and all trees have different
  roots. Then return $x\to \mathit{Tr}_1,\dots,\mathit{Tr}_n$.
\end{itemize}

\begin{example}
  
  Consider the Markov chain below with states $x, y$, where $y$ is a
  terminating state, and $x$ can transition to both $x$ and $y$ with
  probability $\sfrac{1}{2}$.  The smallest fixpoint of the behaviour
  function $\mu \be$ maps $x$ and $y$ to one, i.e.\ both states
  terminate almost surely. We want to witness that the probability for
  $x$ is strictly greater than $\sfrac{1}{2}$.

  \begin{center}
    \begin{tikzpicture}[node distance=1.5 and 1, on grid, shorten
      >=1pt, >=stealth', semithick]

      \begin{scope}[state, inner sep=3pt, minimum size=0pt]
        \draw node [draw] (q0) {\(x\)}; \draw node [draw, right=2cm of
        q0, accepting] (q1) {\(y\)};
      \end{scope}
      
      \begin{scope}[->]
        \draw (q0) edge node [midway, above] {\(\frac{1}{2}\)} (q1);
        \draw (q0) edge[loop] node [midway, above] {\(\frac{1}{2}\)}
        (q0);
      \end{scope}
    \end{tikzpicture}
  \end{center}

  We will denote a function $f: x\mapsto a, y\mapsto b$ by $(a,b)$.
  Basis elements are tuples $f^a_x = (a, 0)$ or $f^b_y = (0,b)$, where
  $a,b>0$.

  Hence we start with the basis element $f^{\sfrac{1}{2}}_x$ and
  observe that $\degr(f^{\sfrac{1}{2}}_x) = 3$, i.e. $i=3$ is the
  first index for which $beh^i((0,0))(x)> \sfrac{1}{2}$ holds.

  \smallskip

  \noindent\emph{Strategy computation:} The computation of the
  strategies involves solving the following inequalities, and thereby
  distributes the required probability of $\sfrac{1}{2}$ to the
  successors of $x$:
  \begin{align*}
    \delta(x)(x)\cdot c_1 + \delta(x)(y)\cdot c_2 = \sfrac{1}{2}\cdot
    c_1 + \sfrac{1}{2}\cdot c_2 > \sfrac{1}{2} \\
    \be^2((0,0))(x) = \sfrac{1}{2} >c_1 \\
    \be^2((0,0))(y) = 1> c_2
  \end{align*}
  Player $\exists$ can choose any $c_1 = \sfrac{1}{2} -\epsilon$,
  $c_2=1-\delta$ ($\epsilon,\delta>0$) that satisfy the first
  inequality. Setting $c_1=0$ or $c_2=0$ is not an option, as then the
  first inequality is not satisfied. Assume that $\exists$ chooses
  $c_1 = \sfrac{1}{4}$ and $c_2 = \sfrac{7}{8}$. Hence the (finitary)
  strategy of $\exists$ is to play $d = \bigsqcup F$ with
  $F = \{f_x^{\sfrac{1}{4}},f_y^{\sfrac{7}{8}}\}$.  Note that
  $\degr(f^{\sfrac{1}{4}}_x) = 2$ and $\degr(f^{\sfrac{7}{8}}_y) = 1$.

  Now, the probability of $c_1 = \sfrac{1}{4}$ needs to be distributed
  to the successors of $x$. Here we assume that $\exists$ chooses
  values $0,\sfrac{7}{8}$, resulting in the strategy
  $F' = \{f_y^{\sfrac{7}{8}}\}$. Finally, for $f_y^{\sfrac{7}{8}}$ it
  is sufficient to play the bottom element $(0,0)$, i.e., the empty
  strategy.

  \smallskip

  \noindent\emph{Playing the game:} 
  Starting with $f^{\sfrac{1}{2}}_x \ll \mu \be$, player $\exists$
  plays $\bigsqcup F = (\sfrac{1}{4},\sfrac{7}{8})$ as determined
  above.  Now $\forall$ has to answer either with
  $(\sfrac{1}{4}-\epsilon,0)$ or $(0,\sfrac{7}{8}-\delta)$. As
  detailed above, $\exists$ can win in the latter case with the empty
  strategy, in the former case $\exists$ plays
  $\bigsqcup F' = (0,\sfrac{7}{8})$, winning the game in the next
  step, since by assumption $\forall$ can not answer with
  $\bot = (0,0)$.

  \smallskip

  \noindent\emph{Constructing the witness:} Now we construct witnesses
  recursively from the strategy. More concretely, we obtain:
  \begin{align*}
    \mathit{wit}(f_y^{\sfrac{7}{8}}) &=
      W_{p,\lo}(f_y^{\sfrac{7}{8}},\emptyset) = y, \qquad \text{(base
        case)} \\ 
    \mathit{wit}(f_x^{\sfrac{1}{4}}) &=
      W_{p,\lo}(f_x^{\sfrac{1}{4}},\mathit{wit}[\{f_y^{\sfrac{7}{8}}\}]) = 
      W_{p,\lo}(f_x^{\sfrac{1}{4}},\{y\}) = 
        x \rightarrow y \\ 
    \mathit{wit}(f_x^{\sfrac{1}{2}}) &=
      W_{p,\lo}(f_x^{\sfrac{1}{2}},\mathit{wit}[\{f_x^{\sfrac{1}{4}},f_y^{\sfrac{7}{8}}\}]) = 
      W_{p,\lo}(f_x^{\sfrac{1}{2}},\{x\rightarrow y, y\}) =
        x \rightarrow(x\rightarrow y), y 
  \end{align*}

  Thus, the termination probability of $x$ being greater than
  $\sfrac{1}{2}$ is witnessed by the tree
  $x \rightarrow(x\rightarrow y), y$, where the subtree $y$ witnesses
  the termination probability of $y$, which is greater than
  $\sfrac{7}{8}$, and the subtree $x\rightarrow y$ witnesses the
  termination probability of $x$ of at least $\sfrac{1}{4}$.  In fact,
  the resulting tree even gives the value $\sfrac{3}{4}$.
\end{example}

\subparagraph*{Dual Case:}~

\noindent\emph{Basis:} For the dual case, a meet basis of $\lt{L}$ is given by the
co-singletons
$\overline{\{\mathit{Tr}\}} = \mathit{Trees}\backslash\{\mathit{Tr}\}$,
and for $\lt{B}$, we define the basis elements to be all functions
$\du{f}^c_x$ ($x\in X$, $c < 1$) where
\[ \du{f}^c_x (y) =
  \begin{cases}
    c & \text{if } x=y \\
    1 & \text{otherwise}
  \end{cases} \]

\noindent\emph{Strategy computation
  (cf. Prop.~\ref{prop:winning-strategy-primal-way-below}):}
$\strat_{d,\be}^\exists(\du{f}^c_x)$ is determined analogous to the
primal case.

\medskip

\noindent\emph{Auxiliary functions
  (cf. Sct.~\ref{sec:auxiliary-functions}):} The auxiliary functions
$\ads$, $\adc$ can also be defined analogously to the dual case. In
addition:
\begin{itemize}
\item $Z(\du{e},\du{d})$: $\du{e},\du{d}\in [0,1]^X$ with
  $\du{e}\not\sqsubseteq\du{d}$, i.e., there exists $x\in X$
  with $\du{e}(x) > \du{d}(x)$. In this case choose $c$
  such that $\du{e}(x) > c > \du{d}(x)$ and return
  $\du{b'} = \du{f}_x^c$.
\end{itemize}

\section{Conclusion}
\label{sec:conclusion}

  We have shown how to generate witnesses, generalizing 
  the construction of distinguishing formulas
  \cite{c:automatically-explaining-bisim,RadyBreugelExplProbBisim} for
  explaining non-bisimilarity or certifying lower bounds. We
  concentrated in particular on guarantees for obtaining finitary
  strategies, which can then be transformed into witnesses. The
  picture below summarizes the paper: on the logic side we have
  both witnesses and strategies $\strat_{p,\lo}^\exists$ proving that
  a witness is generated by the logic function. 
  Proposition~\ref{Wp,lo,E to Wp,be,E} explains how to transform
  witnesses into a winning strategy $\strat_{p,\be}^\exists$ of the
  primal way-below game, while Theorem~\ref{thm:wit-of-StrtEpb} states
  the other direction. Analogously for strategies
  $\strat_{d,\be}^\forall$ of the dual game, where the connection
  arises from Proposition~\ref{Wp,L,E to Wd,B,A} and
  Theorem~\ref{thm:stratFdbToWit}. Existence of finitary strategies in
  the various games is guaranteed by
  Propositions~\ref{prop:winning-strategy-primal-way-below}
  and~\ref{prop:finite-winning-strategy-dual}.

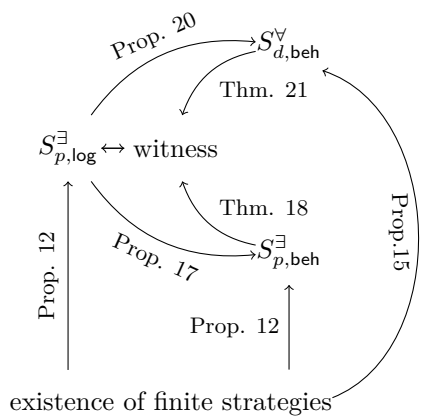
\begin{wrapfigure}{r}{.38\textwidth}
\begin{tikzpicture}[ every state/.style={draw=none, inner sep=0,outer sep=0, rectangle, 
    }]
    \node[state]  (l_0) []                       {$\strat_{p,\lo}^\exists$};
    \node[state]  (l_1) [right=.3cm of l_0]      {\ witness};
    \node[state]  (xx)  [below=2.5cm of l_0]       {}; 
    \node[state]  (b_1) [below right=.7cm of l_1]{$\strat_{p,\be}^\exists$};
    \node[state]  (yy)  [below =1.122cm  of b_1]         {}; 
    \node[state]  (b_2) [above right=.7cm of l_1]{$\strat_{d,\be}^\forall$};
    \node[state]  (b_0) [below=2.5cm of l_1]       {existence of finite strategies};

    \path[->,font=\footnotesize] 
      (l_0) edge [] node [below]  {}   (l_1)
      (l_0) edge [bend left] node [above, sloped]  {Prop. \ref{Wp,L,E to Wd,B,A}}   (b_2) 
      (l_0) edge [bend right] node [below, sloped]  {Prop. \ref{Wp,lo,E to Wp,be,E}}   (b_1) 
      (l_1) edge [] node [below]  {}   (l_0)
      (xx) edge [ ] node [above, sloped]   {Prop. \ref{prop:winning-strategy-primal-way-below}}(l_0) 
      (yy) edge [ ] node [left]   {Prop. \ref{prop:winning-strategy-primal-way-below}}(b_1) 
      (b_0) edge [bend right=70] node [below, sloped] {Prop.\ref{prop:finite-winning-strategy-dual}}(b_2) 
      (b_1) edge [bend left]     node [above right]  {Thm. \ref{thm:wit-of-StrtEpb}}   (l_1)
      (b_2) edge [bend right]    node[below right]   {Thm. \ref{thm:stratFdbToWit}}      (l_1) 
      ;
  \end{tikzpicture}
  \end{wrapfigure}

  In order to obtain formal computability results, it would in
  addition be necessary to make assumptions on the decidability of
  the order relation, computability of joins and auxiliary
  functions, etc. \cite{s:effectively-given-domains}.

  Instantiating this framework to the case of bisimilarity yields a
  game reminiscent of the bisimulation game in \cite{stirling99} and a
  witness construction similar to
  \cite{c:automatically-explaining-bisim}. We also rediscover a known
  construction for behavioural metrics and study a new example in the
  context of Markov chains.

  Another general framework for bisimilarity and behavioural metrics
  is coalgebra \cite{r:universal-coalgebra} in which sound and
  complete logics have been studied
  \cite{s:coalg-logics-limits-beyond-journal,p:coalgebraic-logic},
  giving rise to a Hennessy-Milner theorem. While there is work on the
  construction of distinguishing formulas in a qualitative coalgebraic
  setting \cite{kms:non-bisimilarity-coalgebraic}, there is -- to the
  best of our knowledge -- no general construction in the quantitative
  coalgebraic case. Furthermore our framework offers the flexibility
  of arbitrary lattices.

  The current work could be extended in several directions, such as
  exploring characteristic formulas (a characteristic formula of a
  given state characterizes all states that are in a preorder relation
  to the original state \cite{s:characteristic-formulae}) instead of
  distinguishing formulas or studying the use of the dual game on the
  logic side.

  We also plan to investigate further examples, where witnesses
  generated from the strategies of the primal and dual game might be
  completely different. Potential application areas are dataflow
  analysis and abstract interpretation where least and greatest
  fixpoints play a major role
  \cite{cc:ai-unified-lattice-model,nnh:program-analysis}.

  Furthermore we are interested in the connection to the codensity
  game \cite{kkhkh:codensity-games}. This game uses predicates in the
  ``behaviour game'' and might give rise to games that are
  simultaneously played on both lattices. In general we believe that
  the lattice-based approach can be used to classify and categorize
  various types of behavioural games that have been presented in the
  literature, such as
  \cite{stirling99,vjwb:explainability-probabilistic-game,fkp:expressiveness-prob-modal-logics,dlt:approx-analysis-prob}
  for concrete types of transition systems and
  \cite{km:bisim-games-logics-metric,fmskb:graded-monad-games,fsw:conformance-games-graded,kkhkh:codensity-games}
  in the coalgebraic setting.

  In addition we will further investigate the connection to apartness
  \cite{gj:apartness-bisimulation} in particular to
  \cite{tbbkr:witnesses-lower-bounds-beh-distances}. In the latter
  paper proof systems are used to obtain witnesses for lower bounds
  with the connection to games mentioned as future work.

     \bibliography{references}
    
\end{document}